\documentclass
[preprint,byrevtex,aps,twocolumn,lengthcheck,tightenlines,11pt]{revtex4}%
\usepackage{amsfonts}
\usepackage{amsmath}
\usepackage{amssymb}
\usepackage{graphicx}
\usepackage{float}%
\providecommand{\U}[1]{\protect\rule{.1in}{.1in}}
\newtheorem{theorem}{Theorem}

\newtheorem{proposition}[theorem]{Proposition}

\newenvironment{proof}[1][Proof]{\noindent\textbf{#1.} }{\ \rule{0.5em}{0.5em}}
\begin{document}
\preprint{ }
\title[Death of the classical vibrating string model]{Resolution of the 300-Year-Old Vibrating String Controversy}
\author{Namik Ciblak}
\affiliation{Yeditepe University, Istanbul, Turkey}

\begin{abstract}
The dispute about the well-known 1D vibrating string model and its solutions,
known as "The Vibrating String Controversy", spanned the whole of 1700s and
involved a group of the most eminent scientists of the time. After that, the
model stood undisputed for over two centuries. In this study, it is shown that
not only this 300-year-old model cannot correspond to reality, but it is
theoretically not quite plausible, either. A new 2D model is developed
removing all the assumptions of the classical model. The result is a pair of
non-linear partial differential equations modeling 2D motions of a finite 1D
string. A theorem that can be used to determine the initial displacement
functions from the initial shape of the string is proven. The new model is
capable of representing initial conditions that cannot be handled in the
classical model. It also allows initially non-taut/non-slack strings and
self-intersecting shapes. The classical model and the non-taut strings emerge
as special limit cases. It is proven that pure transverse motions of a 1D
string are possible only in very rare cases. A theorem that sets the
conditions for pure transverse motions is also presented. Numerical studies of
interesting cases are presented in support of the new model. High-speed camera
experiments are also conducted, the results of which also support the new theory.

\end{abstract}
\volumeyear{year}
\volumenumber{number}
\issuenumber{number}
\eid{identifier}
\date{January 2015, first manuscript}
\startpage{1}
\endpage{23}
\maketitle

\section{INTRODUCTION}

About 267 years ago d'Alembert published his results on the vibrating string
problem (1747). He gave the model in the form of a now-well-known partial
differential equation (PDE), namely, the one dimensional wave equation:
$y_{xx}=y_{tt}$. He also found the general solution as $f(x+t)+f(x-t)$,
\cite{Dalembert1}, \cite{Dalembert2}. Following this, Euler published his
results (1748, 1749), in which, using d'Alembert's solution, he showed how to
construct the solution in terms of initial conditions, \cite{Euler}. A few
years later, Daniel Bernoulli presented his work on the subject, claiming that
any solution to the wave equation can be given as an infinite sum of a
fundamental sine function and its harmonics, although he did not know how to
determine the coefficients in the sum, \cite{BernoulliD1753}.

The communications between these three scientists were somewhat bitter and
seemed irreconcilable at the time. D'Alembert didn't like Euler's kinky
initial conditions, Euler argued with infinitesimals, Bernoulli insisted on
trigonometric series, and, neither d'Alembert nor Euler thought that the
trigonometric series were general enough. The matter was not resolved until
1829 when, based on Fourier's seminal work (1819), Dirichlet was able to end
the discussion by presenting his work on the convergence of Fourier series,
\cite{Fourier}, \cite{Dirichlet}. This otherwise fruitful activity in the
history of mathematics and classical physics is known as \emph{The Vibrating
String Controversy}. Detailed accounts of this debate are given in two
excellent articles: one by Wheeler and Crummett, \cite{Wheeler}, and the other
by Zeeman, \cite{Zeeman}.

What was known until the mid 1700s was comprised of the knowledge from the
antiquity (Pythagoreans and successors) and works of earlier contemporaries,
mostly concerning harmonics and the fundamental mode. Then, a slow but
persistent development of the vibrating string theory ensued. Most notable
achievements among these were the following.

\begin{itemize}
\item Joseph Sauveur's experimental work on harmonics (1700), \cite{Zeeman}.

\item Brook Taylor's (of Taylor series) discovery of the shape and the
frequency of the fundamental mode (1713), \cite{Taylor}.

\item Johann Sebastian Bach's work on musical intervals (1722), which he
demonstrated in his well-known collection of compositions called \emph{The
Well-Tempered Clavier}, \cite{Zeeman}.

\item Johann Bernoulli's method of finding the shape of the fundamental mode
from the infinite limit of lumped masses on an elastic string (1727),
\cite{BernoulliJ1732}.

\item Daniel Bernoulli's work on harmonics (1733, 1740), \cite{BernoulliD1732}%
, \cite{BernoulliD1740}.

\item Lagrange's work on sound (1759), \cite{Lagrange}.

\item Fourier's development of the infinite trigonometric series (1819),
\cite{Fourier}.

\item Dirichlet's work on the convergence of Fourier series (1829),
\cite{Dirichlet}.
\end{itemize}

For all who participated in the debate the main concerns were mostly about the
existence, the admissibility, and the generality of the functions proposed as
solutions. Otherwise, all of them, except d'Alembert, considered the model a
sufficiently faithful representation of the actual physical phenomenon,
especially as it pertains to stringed musical instruments.

The classical model of d'Alembert neglects the effects of bending and other
mechanical causes such as damping and rotational motion, which would have
rendered the model non-linear, or increased the degree of the differential
equation, at the least. It also assumes a constant tension, although it does
not introduce any material behavior.

All of these assumptions and omissions, the debaters defended by arguing small
deflections and slopes. Euler actually used the infinitesimal argument to show
the feasibility of admitting solutions with kinks, which left spatial
derivatives at such kinks undefined. Today, it is understood that such
discontinuities can be handled with special functions (distributions) or by
using $C^{\infty}$ approximations (e.g. Fourier series or Legendre polynomials
approximations to the Dirac Delta Function).

Nevertheless, d'Alembert objected to Euler's suggestions by arguing that the
assumption of small deviations from the unloaded equilibrium, which enabled
him to derive the wave equation, precluded any real physical connections. He
was probably trying to indicate that the model was of purely mathematical
nature, nothing to do with the actual physics of a vibrating string.

D'Alembert was both right and wrong. We know that models obtained by such
approximations, usually yielding linear models, can quite accurately predict
the corresponding physical phenomena within limits, as in the cases of linear
theory of elasticity, simple pendulum, linearized constitutive models of
numerous mechanical, electrical, and electromechanical elements, and many
others. So, he was wrong. Yet, he was right, too: the assumptions of the
classical model were so strict that it would only amount to wishful thinking
to expect the real string vibrations to be so described. This point is what we
make in this study, too.

However, more importantly, the classical model makes a very crucial
assumption: \emph{purely transverse motions}. This, no one seemed to contest.
In this study, this assumption will be shown to be indefensible and, further,
that \emph{its removal is the key to settling all other issues}.

Then, again, the question remains: why does the classical string model seem to
fit the experimental observations, anyway? Or, does it, really? There had been
numerous attempts to determine whether or not the classical model fits the
actual string vibrations. An account of these can be found in an article by
Armstead and Karls, \cite{Armstead2006}, in which they also present their own
experimental results. They use a vertical motion model with a variable
tension, based on Kreysig \cite{Kreysig1983} and Powers \cite{Powers}. This
will be another question that we will try to resolve. Indeed, we show by
actual experiments that the classical string model is simply untenable.

\subsection{Framework of the classical model}

We will now scrutinize the assumptions of the classical vibrating string
model. However, before doing this, we have to deal with the important
declaration of the classical model, which one may consider as a rule or
requirement, rather than an assumption. Namely, the requirement that all
motion is to be \textquotedblleft transverse\textquotedblright\ or
\textquotedblleft lateral\textquotedblright, which is the reason that the
problem is sometimes referred to as \textquotedblleft transverse vibrations of
\ldots\textquotedblright\ More correctly, in the classical model the points
are considered to move only in transverse (lateral, perpendicular) directions
\emph{with respect to the unloaded string at rest}.

Nevertheless, we will continue using the same terminology. Thus, if anyone
asks \textquotedblleft\textit{Can a string be made to execute transverse
motions?}\textquotedblright\ We understand that the inquiry is actually about
\textquotedblleft up/down\textquotedblright\ motion with respect to the
initial horizontal configuration.

One can imagine infinitely many constraints on each point of the string that
allow them only move in lateral directions. However, as d'Alembert did, one
also has to admit that such a contraption is almost impossible to achieve in
reality. Further, one would have to redefine the meaning of "free vibrations"
if such constraints are introduced. Nevertheless, we consider this as a
separate class of string motions and, for now, restrict the discussion to
transverse motions.

The basic assumptions of the classical string model with transverse vibrations
are as follows.

\begin{enumerate}
\item \textbf{Perfect Flexibility}. This basically implies that the string
exhibits no reaction to bending. This assumption is necessary to propose that
the only internal reaction to deformation is tension. Any attempt to include
bending will increase the degree of the PDE. Further, without small
deformations assumptions, the bending effect will be non-linear (see
\cite{Gottlieb1990}, \cite{Lai2008}, \cite{Rao2002}). As before, we may
consider this assumption as a matter of classification rather than an approximation.

\item \textbf{Constant Tension}. This is one of the most destructive
assumptions of the classical model because it deceives one into believing that
without it the model would become more complicated. We shall show here that
not only is it completely unnecessary, but it is not justifiable, either. That
is, one cannot develop a string model for transverse motions with a constant
tension assumption. However, once this assumption is removed, it will also be
necessary to remove certain others as explained below. There are studies that
allow variable tension, yet they fall into other traps, eventually resulting
in the classical model.

\item \textbf{Inextensibility}. Some models declare this assumption, the need
for which is difficult to understand simply because it is used nowhere. If the
tension is variable so has to be the length. Some authors who use energy
methods declare a constant length and go on to investigate behavior of the
integral $\int_{0}^{L_{0}}\sqrt{1+y_{x}^{2}}dx$ under the assumption of small
deformations and slopes, without any reference to the obvious contradiction
that the integral is the length of the string in motion and is varying. It
will be shown here that not only are the variable tension and the variable
length considerations necessary, but they are also bound to each other. We
cannot claim one without the other. There are many studies that consider the
string as perfectly elastic, yet they still assume all other assumptions of
the classical model, which, in the end, yields the classical model again (see,
for example, \cite{Weinstock1974}). It will be seen here that they are not
distinguishable from the classical model.

\item \textbf{Small Displacements and Slopes}. These would have been
acceptable assumptions, as if proposing a problem classification, had they not
been used to justify previous unjustifiable assumptions. One can easily accept
these while still allowing variable tension and length. The result would only
be small variations in tension and length. We show here that it is
straightforward to work out the classical model using these assumptions while
still rejecting the constancy of tension and/or length.

\item \textbf{Constancy of Density}. Since the mass of any segment, with an
initial length of $dx$, moves only up and down, the linear density (a
distribution) $dm/dx$ does stay constant in time. This is not an assumption,
but a consequence of the requirement of transverse motion, and of conservation
of mass. Spatial variation, however, is allowable and it does not effect the
final result. Note that the actual linear density is $dm/ds$, where $ds$ is
along the curve of the string, which may not be constant. Classical or other
similar models do not make this distinction between the two definitions of the
linear density, which is justified in the end because of small displacements
and slopes assumption.

\item \textbf{Constancy of Cross-sectional Area}. If the string is idealized
as a curve without thickness, the cross-sectional geometry becomes irrelevant,
which eliminates the issue. However, if a model closer to the reality is
desired then the variability of the area becomes important. Spatial variation
of the area can be handled as in the case of density. We have to point out
that for real strings the case of spatial variation of the area is far more
likely than the density case -- an unavoidable result of real manufacturing processes.
\end{enumerate}

In summary, in this study the assumptions of constant tension,
inextensibility, and, small deformations and slopes are removed. This brings
the model much closer to the reality as well as making it theoretically more
consistent. For example, in any stringed musical instrument both the tension
and the length exhibit significant variations, so much so that the pitch of
the sound may shift audibly from the first plucked moment to the end. Many
virtuoso players use this fact to create beautiful artistic effects. Further,
variation in tension can be so high that the string may simply break.

It is ironic that after more than 300 years, if the results here are of any
merit, we find ourselves back in the middle of that old debate, this time
questioning the model itself, rather than the solutions and the related
mathematical machinery.

\section{LARGE DISPLACEMENT MODEL FOR TRANSVERSE VIBRATIONS}%

\begin{figure}
[ptb]
\begin{center}
\includegraphics[scale=0.3]{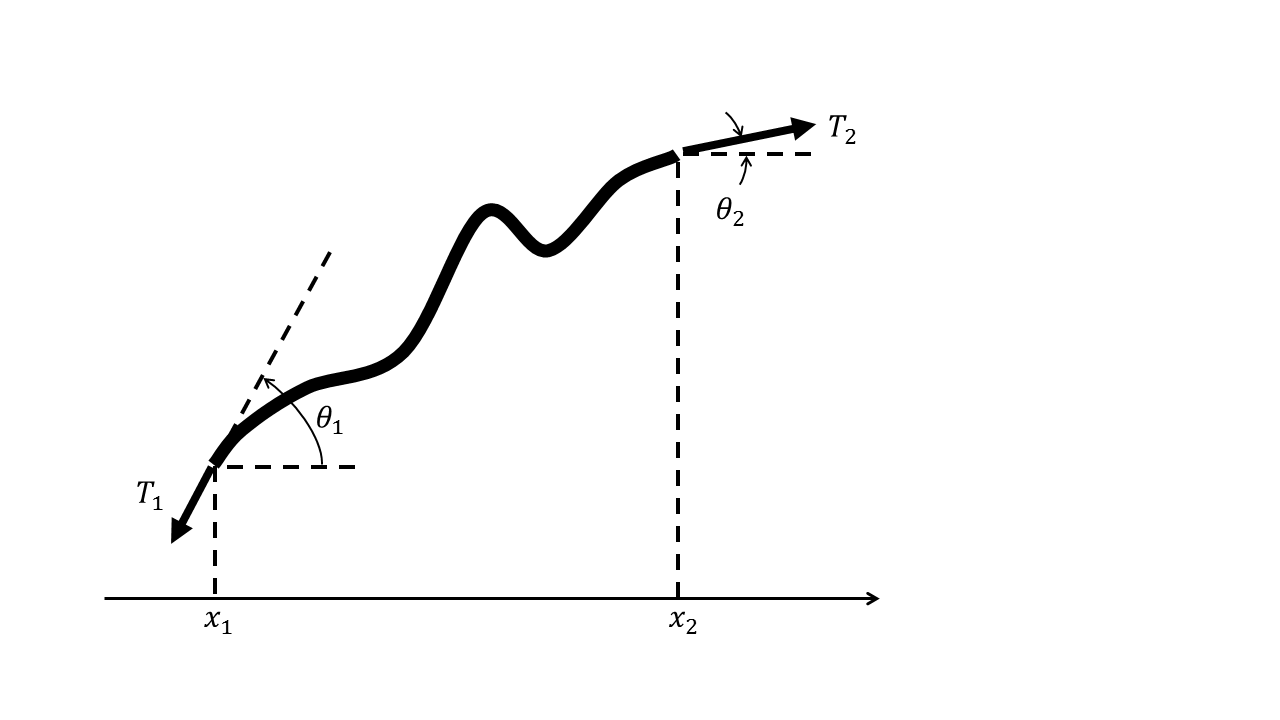}%
\caption{A finite string piece in motion.}%
\label{gen1DMot}%
\end{center}
\end{figure}

Here, we develop a vibrating string model for transverse motions without
arguing limits, infinitesimals, or first order approximations. We start with a
really finite string piece. Let $T_{1}$ and $T_{2}$ be the tensions, and,
$\theta_{1}$ and $\theta_{2}$ be the angles at left and right cuts,
respectively (Figure \ref{gen1DMot}). Then, the force balance in $x$ direction
dictates%
\begin{equation}
T_{1}\cos\theta_{1}=T_{2}\cos\theta_{2}%
\end{equation}
This relation must hold for any segment, hence for any pair of end points
$x_{i}$ and $x_{j}$, at all times. Letting $f(x,t)=T(x,t)\cos(\theta(x,t))$,
the fact that $f(x_{i},t)=f(x_{j},t)$ for all $x_{i}$ and $x_{j}$ leads to
$f(x,t)=g(t)$. Hence, $T(x,t)=g(t)\sec\theta(x,t)$.

However, tension in the string appears as a reaction to changes in local
geometry. Thus, given a shape at any $t^{\ast}$, the tension is dependent on
the shape regardless of $t^{\ast}$. As such, explicit dependence of tension on
time is not admissible. As a result, we propose%
\begin{equation}
T(x,t)=C\sec\theta(x,t)
\end{equation}
where $C$ is a constant.

Note that we have not employed any differentials or approximations in
obtaining this result. Therefore, simply in order to be compatible with
Newton's second law, regardless of such details as whether bending is included
or not, a constitutive model is utilized or not, and so on, and regardless of
how complicated or higher order the model is, any string model must conform to
this result, \textit{provided that only transverse motions are allowed}.

Similar things can be done to the transverse motions. The force balance for
the finite string piece in transverse direction at any time gives%
\begin{equation}
T_{2}\sin\theta_{2}-T_{1}\sin\theta_{1}=y_{tt}(\bar{x})m(x_{1},x_{2})
\end{equation}
where $m(x_{1},x_{2})$ is the mass of the segment from $x_{1}$ to $x_{2}$, and
$y_{tt}(\bar{x})$ is the acceleration of the center of mass ($\bar{x}$) of the
same. Since the mass in $[x_{1},x_{2}]$ is constant, we have%
\begin{equation}
m(x_{1},x_{2})=(x_{2}-x_{1})m/L_{0}=(x_{2}-x_{1})\rho
\end{equation}
where $\rho$ is the linear density at rest. Therefore,%
\begin{equation}
\frac{T_{2}\sin\theta_{2}-T_{1}\sin\theta_{1}}{x_{2}-x_{1}}=\rho y_{tt}%
(\bar{x})
\end{equation}

Letting $x_{1}=x$, $x_{2}=x+h$, $\bar{x}=x_{1}+\alpha\left(  x_{1},h,y\right)
h$, $\alpha\left(  x_{1},h,y\right)  >0$, and taking the limit as
$h\rightarrow0$, yields%
\begin{equation}
\frac{\partial}{\partial x}(T(x,t)\sin\theta(x,t))=\rho y_{tt}(x,t)
\end{equation}
which, with $T(x,t)=C\sec\theta(x,t)$, gives%
\begin{equation}
\frac{\partial\theta}{\partial x}\sec^{2}\theta=\frac{\rho}{C}y_{tt}%
\end{equation}
Now, keeping time fixed and using the fact that $\partial y/\partial
x=\tan\theta$, one has%
\begin{equation}
\frac{\partial^{2}y}{\partial x^{2}}=\frac{\partial\theta}{\partial x}%
\frac{\partial\tan\theta}{\partial\theta}=\frac{\partial\theta}{\partial
x}\sec^{2}\theta
\end{equation}
As a result%
\begin{equation}
y_{xx}=\frac{\rho}{C}y_{tt} \label{EqForPureTrans}%
\end{equation}

This result was obtained by Ciblak (2013) in a more or less similar fashion,
\cite{Ciblak2013}. Equation \ref{EqForPureTrans} is the equation of motion for
a vibrating string under the assumption of purely transverse motions, which is
exactly what d'Alembert discovered. There are significant differences
concerning other things, however. Now,

\begin{enumerate}
\item neither the tension nor the length is constant,

\item neither the displacement nor the slopes are assumed to be small,

\item no first order approximations are applied and,

\item both density and area can vary spatially.
\end{enumerate}

Whether or not this departure from the classical model is significant will be
demonstrated later.

\subsection{Nature of C}

Note that the tension can be written as $T(x,t)=Cds/dx$, where $ds$ is the
length of a segment over $dx$. Then, one may argue that if at some time
$t^{\ast}$ the string has $ds=dx$ everywhere, then the tension would be the
same everywhere, say $T_{0}$, i.e. $C=T_{0}$. The existence of a configuration
in which $ds(t^{\ast})=dx$ everywhere means $y_{x}(x,t^{\ast})=0$. Then,
$y(x,t^{\ast})=0$, due to the boundary conditions. Here, we ignored the
discussion of cases in which $ds=dx$ holds in finite intervals of distinct displacements.

Such configurations, $y(x,t^{\ast})=0$, can always be contrived to exist
(e.g., for $t^{\ast}<0$). If the spring is considered to be unloaded and at
rest prior to the application of initial conditions, then one may safely argue
that $y(x,t)=0$ for all $x$, and $t<0$, (simply consider the fact that this is
a trivial solution of the string PDE, and thus admissible). This also forces
that $y_{x}=0$ in the same domain, giving $T(x,t)=T_{0}$ for all $x$ and
$t<0$. As a result, $T_{0}$ is to be interpreted as the internal tension that
would result had the spring been unloaded and at rest. Note that this is not
equal to the tension in the initial condition.

Also note that if we now claim that the variation of $T(x,t)=T_{0}\sec\theta$
is not negligible, then the same is probably true for the total length%
\begin{align}
L\left(  t\right)   &  =\int_{0}^{L_{0}}\sqrt{1+y_{x}^{2}}dx\nonumber\\
L\left(  t\right)   &  =\int_{0}^{L_{0}}\sec\theta dx=\frac{1}{T_{0}}\int
_{0}^{L_{0}}T\left(  x,t\right)  dx
\end{align}

\subsection{What really happens at $t=0$?}

Another counter-intuitive conclusion pertains to the situation at initial
condition. For any motion to ensue the string must be given an initial
displacement or velocity, or both, which have to be achieved by suitable
initial loads. Let's consider a simple initial displacement as shown in Figure
\ref{snglKinkVertMot} with zero initial velocity.%
\begin{figure}
[ptb]
\begin{center}
\includegraphics[scale=0.28]%
{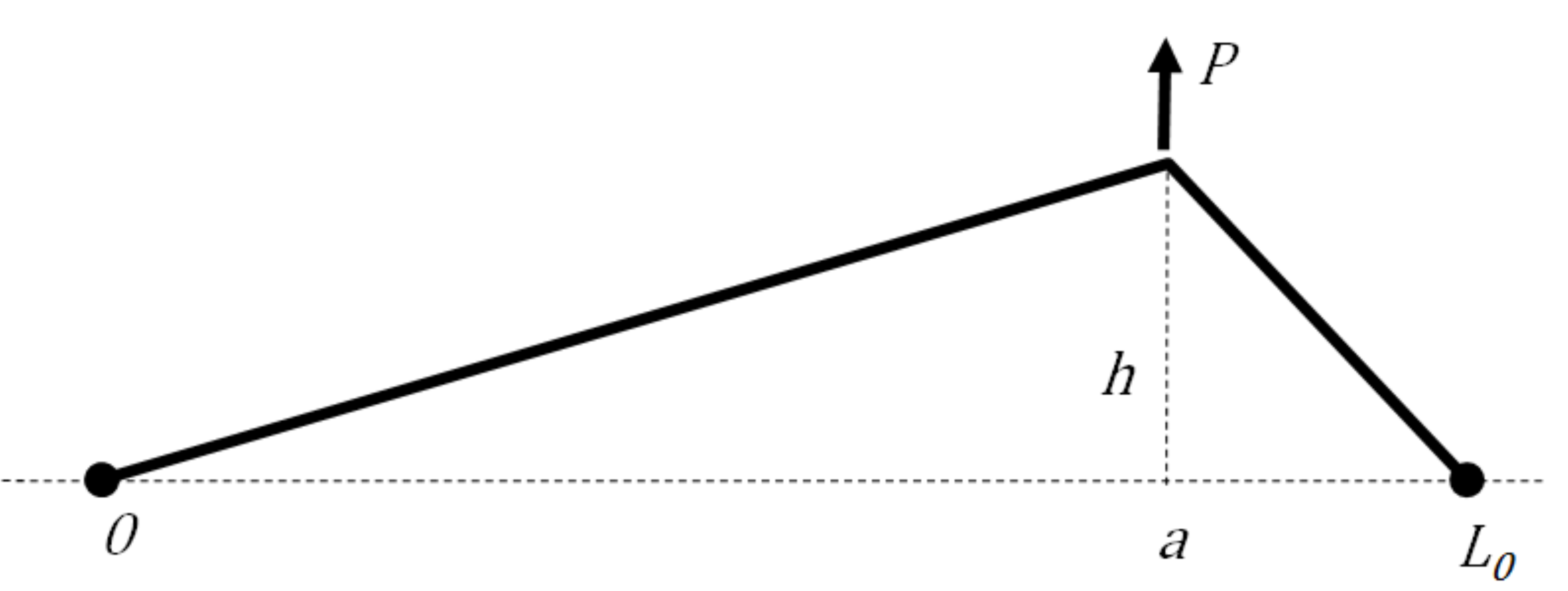}%
\caption{Single kink initial condition with vertical displacement only.}%
\label{snglKinkVertMot}%
\end{center}
\end{figure}

The force $P$ is what is needed to induce the shown displacement. Based on the
tension model presented before, the tensions at any point to the left and
right of the external force are%
\begin{align}
T_{L}  &  =T_{0}\sec\alpha=T_{0}\sqrt{1+(h/a)^{2}}\\
T_{R}  &  =T_{0}\sec\beta=T_{0}\sqrt{1+(h/\left(  L_{0}-a\right)  )^{2}}%
\end{align}

The counter-intuitive point is that these can be quite different, depending on
where the force is applied. Why is this so? The answer is simple: it is the
result of the requirement that the points move in vertical directions only.

Further, from the force balance in vertical direction at point $a$ one gets%
\begin{align}
P  &  =T_{L}\sin\alpha+T_{R}\sin\beta\\
P  &  =T_{0}(\tan\alpha+\tan\beta)=\frac{T_{0}L_{0}}{a(L_{0}-a)}h
\end{align}
This equation indicates that the string behaves like a linear spring with an
equivalent spring stiffness of%
\begin{equation}
k_{a}=\frac{T_{0}L_{0}}{a(L_{0}-a)}%
\end{equation}
Thus, in response to vertical forces, current string model responds like a
linear spring, softest at the mid-point ($k_{\min}=4T_{0}/L_{0}$) and
stiffening, without bound, as $a$ approaches to the boundaries. This behavior
is quite familiar to those who play a plucked string instrument such as guitar.

In any region where the slope vanishes the tension is simply equal to the
tension at rest. This could also be the case in the initial condition. This is
so even though the tension at neighboring points just outside such portions
can differ by a finite amount. For example, if there were two equal forces in
the previous figure applied at $a=L_{0}/3$ and $b=2L_{0}/3$ then, as the
symmetry would require, the slope within the middle section would have been
zero and the tension therein would have to be equal to $T_{0}$.

Again, this behavior is due to the assumption that the string points are
allowed to move only in transverse directions. The tension in the mid-section
would stay constant because there would be no extension in that section. In
reality, however, the tension in the mid-section would increase because some
material would leave the region at both ends due to higher tensions in the
first and third sections. Thus, in an analysis involving real material
behavior the points must be allowed to move in all directions. We shall return
to this later.

\section{UNINTENDED DECEPTION}

Up to here, an unsuspecting reader reads with intent and some level of
scrutiny, not knowing that he or she was deceived into thinking that
everything was fine, equations checked, results made sense, and so on. Yet,
the deception is there.

The author of this paper, as a victim of his own writing, also fell for the
deception, recovering only recently, after about a year since the original
manuscript, \cite{Ciblak2013}. That it is an unintended deception does not
change the fact that it is a quite powerful one. It is discovered only after
the following simple question is posed: \emph{what prevents the points of the
string from moving horizontally?}

From a kinematic point of view, one has the freedom to constrain any point or
body to move in certain ways, as long as this does not violate more basic
requirements such as continuity. However, this freedom can be bought only in
exchange for a constraint force. If a point moves on a circular path, there
must be a force that makes it do so. This quite well-known fact was somehow
overlooked in both the classical models and the model developed in the
previous section, though the latter is intentionally included here only to
make a point.

How did these models work, then, if they even did? The classical model escapes
this question by invoking the smallness arguments. Whereas, the model
developed in previous section does it by tasking the tension to do the
balancing. Whether or not this is acceptable we shall show later. For now, the
question is how to overcome this obstacle if we still insist on having a
transversely vibrating string. We should point out that the material cannot
know whether its boundaries or points are constrained or not. Thus, solving
the problem by using the material in disguise is meaningless.%
\begin{figure}
[ptb]
\begin{center}
\includegraphics[scale=0.55]%
{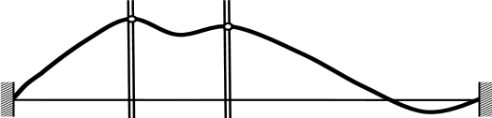}%
\caption{Constraining horizontal motions using vertical rails or slots.}%
\label{twoRails}%
\end{center}
\end{figure}

Figure \ref{twoRails} shows a string, two points of which are constrained to
move in two vertical slots. We can imagine this to be done for all points of
the string. The final result will be a horizontal constraint force
distribution that will be responsible for inducing transverse motions.
However, this would destroy our basic assumption: \emph{the free vibration},
or we would have to redefine it.

Such a horizontal force distribution $q(x,t)$ would be given by%
\begin{equation}
q(x,t)=-\frac{\partial}{\partial x}(T\cos\theta)
\end{equation}

The classical model assumes $\cos\theta\approx1$ leading to $T=T_{0}$, then
yielding $q\left(  x,t\right)  =0$; whereas the model of the previous section
takes $T=T_{0}\sec\theta$, giving $q\left(  x,t\right)  =-\frac{\partial
}{\partial x}(T_{0}\sec\theta\cos\theta)=0$, too. Since the constancy of $T$
is not defendable, we are left with the latter choice, which was developed on
the assumption of $q\left(  x,t\right)  =0$, anyway (the unintended deception).

Thus, our desire to have zero horizontal force distribution determines the
tension model for the material, which is only slightly less than magic.
Nevertheless, this latter model is really closer to reality, from whence it
derives its power, due to the fact that a simple material model reduces to it
in a certain limit (see the next section).

\subsection{Material model}

Instead of the ad hoc material models of previous theories, we now introduce a
well-known material model based on Hooke's law and the linear theory of
elasticity. Obviously, one can choose other material models. However, we
choose a Hookean model for simplicity and demonstration. Similarly, it is also
perfectly allowable to adopt some non-linear theory of elasticity. Considering
the cases of finite strain and effects of strain rate, which can actually be
quite high, this would have been more prudent. Yet, it would take the focus
away from what this study is trying to achieve.

We adopt the definitions of the engineering stress and strain, and the modulus
of elasticity, $E$. Let $A_{f}$ be the cross-sectional area of the string when
it is non-taut and free, and, $dx_{f}$ be the length of a piece whose length
in taut but unloaded case is $dx$. The engineering strain and stress, defined
as $\varepsilon=ds/dx_{f}-1$ and $\sigma=T/A_{f}$, respectively, are related
as $\sigma=E\varepsilon$. Therefore, one can write%
\begin{align}
T/A_{f}  &  =E(ds/dx_{f}-1)\\
T_{0}/A_{f}  &  =E(dx/dx_{f}-1)
\end{align}
which, after eliminations, gives%
\begin{equation}
T=T_{0}\sec\theta+T_{f}(\sec\theta-1) \label{TensionModel}%
\end{equation}
where $T_{f}=EA_{f}$. Note that $T_{f}$ can be considered as a material constant.

The first term of Equation \ref{TensionModel} is the culprit responsible for
the deception, with the aid of the second term, of course. For we can now
easily see that for very small slopes, the equation reduces to $T=T_{0}$, to
that of the classical theory; whereas for small, but not that small slopes, or
for $T_{f}=0$, it reduces to that of the model developed here previously,
because the second term, $\sec\theta-1$, becomes negligible when compared to
the first, $\sec\theta$. The following table summarizes these results.

\begin{center}%
\begin{tabular}
[c]{|l|l|}\hline
Slope & Tension Model\\\hline
\multicolumn{1}{|r|}{General:} & $\left(  T_{0}+T_{f}\right)  \sec\theta
-T_{f}$\\
\multicolumn{1}{|r|}{Small (or $T_{f}=0$):} & $T_{0}\sec\theta$\\
\multicolumn{1}{|r|}{Very Small:} & $T_{0}$\\\hline
\end{tabular}
\smallskip
\end{center}

Note that the general case is now applicable to strings with zero thickness,
too. As to why the linear elastic material model reduces to the classical
model may be attributed to the use of lumped masses connected via springs in
modeling the string in the latter. However, the reason that it also reduces to
the model developed in the previous section remains an open question since we
did not invoke any elasticity condition there. This could be taken as a
testimony for the ubiquity and success of the Hooke's law.

The horizontal force distribution needed to induce purely transverse motions
can now be given as%
\begin{equation}
q(x)=-\frac{T_{f}y_{x}y_{xx}}{\left(  1+y_{x}^{2}\right)  ^{3/2}}=-T_{f}%
\kappa(x)y_{x}\label{HorForceTransVib}%
\end{equation}
where $\kappa(x)$ is the curvature. Note that in regions where the string is
straight the horizontal load vanishes. This property will be shown to result
in important implications.

\subsection{True equation of motion for transverse vibrations}

Using the new tension relation developed above, the vertical force balance
yields the following.%
\begin{align}
\left(  1-k\frac{1}{\left(  1+y_{x}^{2}\right)  ^{3/2}}\right)  y_{xx}  &
=\frac{1}{c_{f0}^{2}}y_{tt}\label{TrueEqTransVib}\\
y_{xx}-k\kappa &  =\frac{1}{c_{f0}^{2}}y_{tt} \label{TrueEqTransVib2}%
\end{align}
where $k=\frac{T_{f}}{T_{0}+T_{f}}$, $c_{f0}=\sqrt{\frac{T_{0}+T_{f}}{\rho}}$,
and $\kappa$ is the curvature. Equation \ref{TrueEqTransVib2} is the correct
equation of motion for the transverse vibrations of a string under a
horizontal force constraint.

Again, we investigate what happens in certain limits. The results are
summarized in the following table.

\begin{center}%
\begin{tabular}
[c]{|l|l|}\hline
Slope & Transverse Vibration Model\\\hline
\multicolumn{1}{|r|}{General:} & $\left(  1-k\frac{1}{\left(  1+y_{x}%
^{2}\right)  ^{3/2}}\right)  y_{xx}=\frac{1}{c_{f0}^{2}}y_{tt}$\\
\multicolumn{1}{|r|}{Small:} & $(1+\frac{3}{2}\frac{T_{f}}{T_{0}}y_{x}%
^{2})y_{xx}=\frac{1}{c^{2}}y_{tt}$\\
\multicolumn{1}{|r|}{%
\begin{tabular}
[c]{r}%
Very Small\\
or $k=0$:
\end{tabular}
} & $y_{xx}=\frac{1}{c^{2}}y_{tt}$\\\hline
\end{tabular}

\end{center}

where $c=\sqrt{\frac{T_{0}}{\rho}}$. Again, the classical model creeps in as a
limit case. Note that for $k\rightarrow0$, either $T_{f}\rightarrow0$ or
$T_{0}\gg T_{f}$.

\subsection{Non-taut strings can vibrate, too}

In classical treatments, it is as if the case of an initially non-taut string
does not exist. Yet, it is a quite familiar reality. What is the equation of
motion governing such a string, then? The classical model is completely silent
on this question. The model derived above, though still unrealistic, allows
one treat initially non-taut, but non-slack, strings. All one has to do is to
set $T_{0}=0$, which corresponds to $k=1$. The result is $T=T_{f}(\sec
\theta-1)$ and%
\begin{equation}
(1-\frac{1}{\left(  1+y_{x}^{2}\right)  ^{3/2}})y_{xx}=\frac{\rho}{T_{f}%
}y_{tt}=\frac{1}{c_{f}^{2}}y_{tt}%
\end{equation}%
\begin{equation}
y_{xx}-\kappa(x)=\frac{1}{c_{f}^{2}}y_{tt}%
\end{equation}
where $\kappa(x)$ is the curvature and $c_{f}=\sqrt{\frac{T_{f}}{\rho}}$.

For small slopes, this equation reduces to $y_{tt}=0$, solution of which is
$y=v_{0}(x)t+y_{0}(x)$, including the initial conditions. For any fixed
boundary conditions, this would yield a non-moving string. For other types, it
would yield non-oscillating solutions moving to infinity. Hence, the classical
limit of non-taut strings is not meaningful.

However, for the small, but not that small, limit one has:%
\begin{align}
\frac{3}{2}y_{x}^{2}y_{xx}  &  =\frac{1}{c_{f}^{2}}y_{tt}\\
\frac{1}{2}\left(  y_{x}^{3}\right)  _{x}  &  =\frac{1}{c_{f}^{2}}y_{tt}%
\end{align}
Letting $u=y_{x}$, taking the derivatives of both sides with respect to $x$,
and using the substitution $\tau=c_{f}t/\sqrt{2}$ gives,%
\begin{equation}
(u^{3})_{xx}=u_{\tau\tau}%
\end{equation}
This is an interesting equation in itself, similar to Burger's equation.

\subsection{Another big oversight}

The problems of the classical string models do not end there. Assuming that
the proper horizontal force distribution is somewhat maintained during motion
so that transverse vibrations are guaranteed, there is still the problem of
initial conditions. This discussion can probably be taken as a tribute to what
Euler was trying to do.

We shall only concentrate on the initial displacement with zero initial
velocity cases. Consider the two distinct initial configurations in Figure
\ref{snglKink2Cases}, which seemingly have identical shapes.%
\begin{figure}
[ptb]
\begin{center}
\includegraphics[scale=0.38]%
{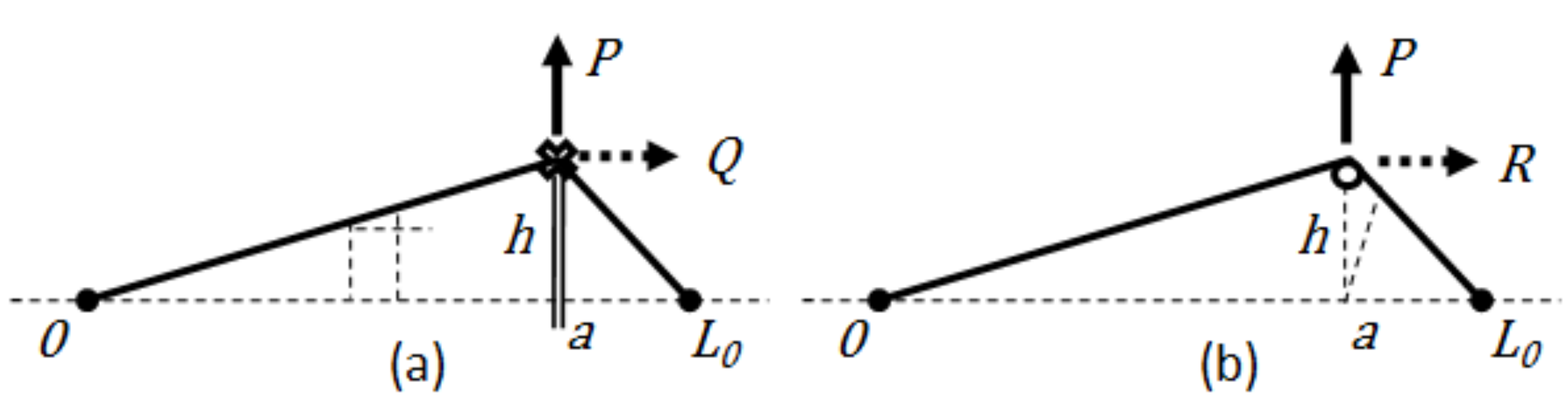}%
\caption{Strings with identical initial shapes, but with distinct displacement
fields.}%
\label{snglKink2Cases}%
\end{center}
\end{figure}

In Figure \ref{snglKink2Cases}(a), the point of force application is
constrained to move along a vertical line, indicated by a linear slot. In this
case, the tensions to the left and right of the force are different, in
general. Every point to the left of the force experience the same strain thus
move vertically without the need for a horizontal force distribution
constraint. Same is true for the right section. This can also be found from
Equation \ref{HorForceTransVib}, in which the vanishing second derivatives
would result in zero force distributions in each region. The force $Q$
represents the action of the vertical slot. Now, the tensions in left and
right regions, $T_{L}$ and $T_{R}$, respectively, are given by%
\begin{equation}
T_{L}=(T_{0}+T_{f})\sqrt{1+\left(  \frac{h}{a}\right)  ^{2}}-T_{f}%
\end{equation}

\begin{equation}
T_{R}=(T_{0}+T_{f})\sqrt{1+\left(  \frac{h}{L_{0}-a}\right)  ^{2}}-T_{f}%
\end{equation}

The horizontal constraint load is given by%
\begin{equation}
Q=T_{f}\left(  \frac{1}{\sqrt{1+\left(  \frac{h}{a}\right)  ^{2}}}-\frac
{1}{\sqrt{1+\left(  \frac{h}{L_{0}-a}\right)  ^{2}}}\right)
\end{equation}
which is curiously independent of the initial tension, and vanishes if
$T_{f}=0$. This may be one of the reasons of the apparent success of the
classical models.

The vertical force is related to the displacement at $a$ as follows.%
\begin{align}
P  &  =(T_{0}+T_{f})\frac{L_{0}h}{(L_{0}-a)a}\nonumber\\
&  -T_{f}\left(  \frac{h}{\sqrt{(L-a)^{2}+h^{2}}}+\frac{h}{\sqrt{a^{2}+h^{2}}%
}\right)
\end{align}
Now, the linear dependence of $P$ on $h$ is lost, although it is recovered for
sufficiently small $h$ or $T_{f}=0$.

The case in Figure \ref{snglKink2Cases}(b) is more problematic. The vertical
force is applied to a negligibly small and frictionless pulley that allows the
tension on both sides to balance. The horizontal force $R$ is responsible for
keeping the pulley in a vertical path. In the case shown, i.e. when the point
of force application is to the right of the midpoint of the string, the
tension would tend to be higher in the right region causing some material to
flow from the left to the right, which results in relaxation towards a
balance. For example, the material point that was initially at $x=a$ would
follow an oblique line into the right region. As a result, despite such a
perfectly allowable initial configuration, none of the previous models would
be applicable because we now have non-transverse motion or, more correctly, a
general motion in $xy$-plane. Note that in this case, when the tension is
allowed to balance, the relation $\frac{ds}{d\theta}=\sec\theta$ does not hold
any longer.

There are more complicated, yet quite proper, initial configurations that
cannot be handled by our simplified models. Figure \ref{ICwithHorMot} shows
only two of such situations which are quite easily doable, using a simple
rubber string, for example. Again, Euler's assertion, that any initial shape
that can be drawn by hand without lifting the pen should be admissible, is
true except that he, for some reason, did not consider such simple initial
configurations as in Figure \ref{ICwithHorMot}, which are impossible to handle
using d'Alembert's model.

Obviously, for large displacements one would most likely encounter collisions
or self-intersecting shapes. These aside, however, we do not even have a model
for small displacements allowing 2D motions or one that can handle simple
initial conditions as in Figure \ref{ICwithHorMot}. All of these point to a
need for a more general vibration model for the one dimensional string. One
such model is presented in the sequel.%
\begin{figure}
[ptb]
\begin{center}
\includegraphics[scale=0.38]%
{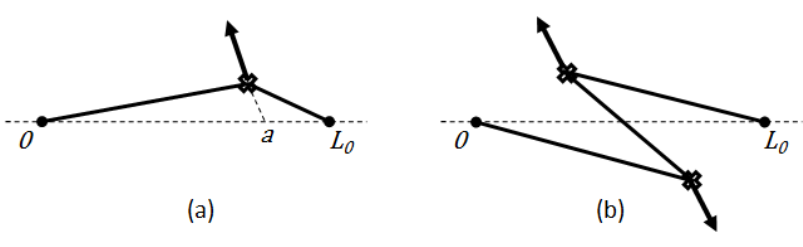}%
\caption{Initial conditions with forced horizontal displacements.}%
\label{ICwithHorMot}%
\end{center}
\end{figure}

\section{2D VIBRATIONS OF 1D STRING}%

\begin{figure}
[ptb]
\begin{center}
\includegraphics[scale=0.4]%
{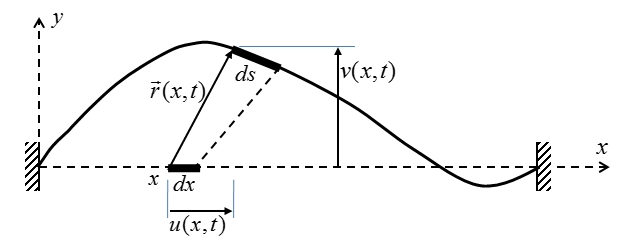}%
\caption{String in 2D motion.}%
\label{gen2DMot}%
\end{center}
\end{figure}

Figure \ref{gen2DMot} shows a string in general plane motion in which an
infinitesimal element originally at $x$, with a length of $dx$, is moved to
$(x+u(x,t),v(x,t))$ and stretched to $ds$, where $u(x,t)$ and $v(x,t)$ are the
displacement components along the $x$ and $y$ coordinate axes, respectively.
The displacement components form scalar fields over the span of $x$ and, for
now, we consider them to be at least twice differentiable functions of $x$ and
$t$.

The stretched length of the infinitesimal element is given by%
\begin{equation}
ds=\sqrt{(1+u_{x})^{2}+v_{x}^{2}}dx
\end{equation}
Now, however, $\sec\theta\neq ds/dx$. Nevertheless, the general model for the
tension still applies as: $T=(T_{0}+T_{f})ds/dx-T_{f},$ therefore,%
\begin{equation}
T=(T_{0}+T_{f})\sqrt{(1+u_{x})^{2}+v_{x}^{2}}-T_{f}%
\end{equation}
The force balance equations remain simple:%
\begin{align}
u_{tt}dm  &  =d(T\cos\theta)\\
v_{tt}dm  &  =d(T\sin\theta)
\end{align}
Now, we use the following.
\begin{subequations}
\label{cossin}%
\begin{align}
\cos\theta &  =\frac{1+u_{x}}{\sqrt{(1+u_{x})^{2}+v_{x}^{2}}}\\
\sin\theta &  =\frac{v_{x}}{\sqrt{(1+u_{x})^{2}+v_{x}^{2}}}%
\end{align}
and $dm=\rho_{0}dx$, where $\rho_{0}$ is the linear density at rest, to get
the final result presented below.%
\end{subequations}
\begin{align}
\frac{1}{c_{f0}^{2}}u_{tt}  &  =u_{xx}-k\frac{\partial}{\partial x}%
\frac{1+u_{x}}{\sqrt{(1+u_{x})^{2}+v_{x}^{2}}}\label{Main2DEqU}\\
\frac{1}{c_{f0}^{2}}v_{tt}  &  =v_{xx}-k\frac{\partial}{\partial x}\frac
{v_{x}}{\sqrt{(1+u_{x})^{2}+v_{x}^{2}}} \label{Main2DEqV}%
\end{align}
where $k=\frac{T_{f}}{T_{0}+T_{f}}$ and $c_{f0}=\sqrt{\frac{T_{0}+T_{f}}%
{\rho_{0}}}$.

Equations \ref{Main2DEqU} and \ref{Main2DEqV} are the equations of motion for
the one-dimensional string in $xy$-plane and are the main results of this study.

As a justification, we present below some limit cases, in which $k=0$ cases
correspond to $T_{f}=0$ or $T_{0}>>T_{f}$, whereas $k=1$ cases correspond to
$T_{0}=0$ or $T_{0}<<T_{f}$. The definitions $c^{2}=T_{0}/\rho_{0}$,
$c_{f}^{2}=T_{f}/\rho_{0}$, $c_{f0}^{2}=\left(  T_{f}+T_{0}\right)  /\rho_{0}$
are used throughout. Also note that, for a string with a vanishing thickness
one has $T_{f}=0$, leading to $k=0$.

\subsection{CASE: Transverse motion only}

In this mode $u=0$ and Equation \ref{Main2DEqU} can be satisfied only if $k=0$
(since $v_{x}\neq0$, for motion). Therefore, for transverse motions one must
have $T_{f}=0$. This is true only for a vanishingly thin string $(A_{f}=0)$ or
a material with no resistance to stretching $(E=0)$. The case of very taut
string $(T_{0}>>T_{f})$ also approximates this mode. With $T_{f}=0$ the
tension, the horizontal constraint force, and the remaining equation of motion
become%
\begin{align}
T  &  =T_{0}\sqrt{1+v_{x}^{2}}=T_{0}\sec\theta\label{tensVertMot1}\\
q(x,t)  &  =-\frac{\partial}{\partial x}(T\frac{1}{\sqrt{1+v_{x}^{2}}%
})=0\label{tensVertMot2}\\
\frac{1}{c_{f0}^{2}}v_{tt}  &  =\frac{1}{c^{2}}v_{tt}=v_{xx}%
\end{align}

This is where lies the reason of the apparent success of the classical string
theory, as it is modified by Ciblak, \cite{Ciblak2013}. This time, however,
the tension model based on a simple linear elastic behavior seems to be
exactly what is needed to rescue the classical string model. Nevertheless, in
following sections we show that pure transverse motions are possible only in
rare cases.

Note that pure transverse motions are possible using kinematic constrains such
as a discrete vertical rail systems. However, these introduce hidden boundary
conditions which are not declared anywhere, as well as making the problem look
like an unnatural contraption without any justification. In such cases, the
horizontal constraint force vanishes in between the rails. Therefore,
kinematic constraints such as discrete rails require discrete, horizontal
point forces.

In the rest of this study we only explore the admissibility of pure transverse
motions without any kinematic constraints. In other words, we look for cases
in which $q\left(  x,0\right)  $ is whatever needed to induce pure transverse
displacements at $t=0$, yet $q\left(  x,t\right)  =0$ ($t>0$) despite having
truly pure transverse motions without any kinematic constraints.

\subsection{CASE: Horizontal motion only}

The conditions necessary to induce this mode are: $v=0$ and $\theta=0$,
leading to $ds/dx=1+u_{x},$ and $w(x)=0$ (vertical constraint force). The
equations of motion become

\begin{center}%
\begin{tabular}
[c]{|l|l|}\hline
Limit & Equations\\\hline
General: & $u_{tt}=c_{f0}^{2}u_{xx}$\\\hline%
\begin{tabular}
[c]{l}%
$k=0$:\\
(A material model)
\end{tabular}
& $u_{tt}=c^{2}u_{xx}$\\\hline%
\begin{tabular}
[c]{l}%
$k=1$:\\
(Non-taut string)
\end{tabular}
& $u_{tt}=c_{f}^{2}u_{xx}$\\\hline
\end{tabular}

\end{center}

with the following tension models.

\begin{center}%
\begin{tabular}
[c]{|l|l|}\hline
Limit & Tension\\\hline
General: & $T=T_{0}\left(  1+u_{x}\right)  +T_{f}u_{x}$\\\hline%
\begin{tabular}
[c]{l}%
$k=0$:\\
(A material model)
\end{tabular}
& $T=T_{0}\left(  1+u_{x}\right)  $\\\hline%
\begin{tabular}
[c]{l}%
$k=1$:\\
(Non-taut string)
\end{tabular}
& $T=T_{f}u_{x}$\\\hline
\end{tabular}

\end{center}

Interestingly, the pure horizontal motion case does not suffer from any
troubles similar to what engulfed the pure transverse motions. This is because
the pure horizontal motion, including the initial conditions, does not affect
the vertical displacement.

\subsection{CASE: $k=0$ (a material model)}

In this case the tension becomes $T=T_{0}\sqrt{(1+u_{x})^{2}+v_{x}^{2}}$, and
the equations reduce to%
\begin{equation}%
\begin{tabular}
[c]{l}%
$u_{tt}=c^{2}u_{xx}$\\
$v_{tt}=c^{2}v_{xx}$%
\end{tabular}
\end{equation}

\subsection{CASE: $k=1$ (initially non-taut/non-slack string)}

Now, $T=T_{f}\left(  \sqrt{(1+u_{x})^{2}+v_{x}^{2}}-1\right)  $ and the
equations reduce to%
\begin{equation}%
\begin{tabular}
[c]{l}%
$\frac{1}{c_{f}^{2}}u_{tt}=u_{xx}-\frac{\partial}{\partial x}\frac{1+u_{x}%
}{\sqrt{(1+u_{x})^{2}+v_{x}^{2}}}$\\
$\frac{1}{c_{f}^{2}}v_{tt}=v_{xx}-\frac{\partial}{\partial x}\frac{v_{x}%
}{\sqrt{(1+u_{x})^{2}+v_{x}^{2}}}$%
\end{tabular}
\end{equation}

It can be seen that the classical model and the non-taut model are the two
distinct limits of the general 2D motions of the 1D string.

\subsection{Vector form and extension to higher dimensions}

Using a substitution $\mathbf{r}=[1+u_{x},v_{x}]^{T}\in R^{2}-\{\mathbf{0}\}$,
a vector form of 2D vibrations,%
\begin{equation}
\frac{1}{c_{f0}^{2}}\mathbf{r}_{tt}=\mathbf{r}_{xx}-k\frac{\partial^{2}%
}{\partial x^{2}}\left(  \frac{\mathbf{r}}{\left\Vert \mathbf{r}\right\Vert
}\right)  \label{vecform}%
\end{equation}
is obtained, which can be used as a basis for extending the result to $n$
dimensional motions of a 1D string. That is, by letting $\mathbf{r}%
=[1+u_{x},v_{1x},v_{2x},...,v_{(n-1)x}]^{T}\in R^{n}-\{\mathbf{0}\}$, Equation
\ref{vecform} can be interpreted as describing the motion of a 1D string in
$n$ dimensional space, where $v_{ix}$ are the displacement gradients in
directions perpendicular to $x$.

\section{INITIAL CONDITIONS FOR 2D MOTIONS}

We now return to the question of initial conditions, which is not a
straightforward problem as was seen in previous sections. In general, the
initial spatial configuration must satisfy the following.
\begin{subequations}
\label{ICforces}%
\begin{align}
q(x)  &  =-(T_{0}+T_{f})\frac{\partial}{\partial x}\left(  U_{0x}%
-k\frac{U_{0x}}{\sqrt{U_{0x}^{2}+v_{0x}^{2}}}\right) \\
w(x)  &  =-(T_{0}+T_{f})\frac{\partial}{\partial x}\left(  v_{0x}%
-k\frac{v_{0x}}{\sqrt{U_{0x}^{2}+v_{0x}^{2}}}\right)
\end{align}
where $U_{0x}=1+u_{0x}$ and, $u_{0x}$ and $v_{0x}$ are displacement gradients
at $t=0$, and, $q$ and $w$ are the horizontal and vertical force
distributions, respectively. If the force distributions are given, equations
\ref{ICforces} can be used to determine the initial displacement gradients.
For example, let $Q(x)=\int q(x)dx$ and $W(x)=\int w(x)dx$. Then, the
solutions are%
\end{subequations}
\begin{align}
u_{0}  &  =\int\left(  1+\frac{k}{\sqrt{A_{1}^{2}+A_{2}^{2}}}\right)
A_{1}dx+c_{3}-x\\
v_{0}  &  =\int\left(  1+\frac{k}{\sqrt{A_{1}^{2}+A_{2}^{2}}}\right)
A_{2}dx+c_{4}\\
\left(  \frac{ds}{dx}\right)  _{0}  &  =k+\sqrt{A_{1}^{2}+A_{2}^{2}}\\
T  &  =(T_{0}+T_{f})\sqrt{A_{1}^{2}+A_{2}^{2}}%
\end{align}
where $A_{1}=c_{1}-\frac{Q(x)}{T_{0}+T_{f}}$ and $A_{2}=c_{2}-\frac
{W(x)}{T_{0}+T_{f}}$, and $c_{i}$ are arbitrary constants, which should be
determined from the boundary conditions on $u_{0}$ and $v_{0}$. A useful
special case is summarized in the following proposition.

\begin{proposition}
\label{prop1}A segment of the string in the initial condition is free of
external forces if and only if the displacement gradients (or, equivalently,
the extension ratio or the tension) are constant therein.
\end{proposition}

\begin{proof}
If a segment if free of external forces then equations \ref{ICforces} yield,
after an integration, two algebraic equations of the gradients, whose solution
can be shown to uniquely exist. If the gradients are constant in a segment
then the same equations give zero external force distributions. Note that in
such a segment both the tension and the extension ratio remain constant by
virtue of their definitions.
\end{proof}

If, instead of the external forces, the initial shape is given, then the
situation is more involved. Now, equations \ref{ICforces} can only be used to
determine the required force distributions. Thus, we have to find other ways
of determining the initial displacement fields.

However, as discussed earlier, the initial shape alone is not enough to
determine the initial situation. There could be infinitely many initial
configurations that yield the same shape. Further, in contradiction to what
Euler stated, there are initial shapes that can be drawn by hand without
lifting the pen, but cannot be represented by a simple function. An example of
this was given in Figure \ref{ICwithHorMot}(b). This issue will be addressed later.

Assuming now that the initial shape can be given as $y_{0}(x)$, we distinguish
two major categories: a) those in which the string is allowed to initially
equalize the internal tension, b) and those in which it is not. The latter
requires the description of each particular case as well as introducing extra
boundary conditions. Therefore, we concentrate on the former\ by stating the
following theorem.

\begin{theorem}
\label{theo1}Let the initial shape of the string be given by a continuous and
piecewise differentiable function $y_{0}(x)$ that satisfies the boundary
conditions. If the tension is the same everywhere initially, then the initial
displacement fields are given by%
\begin{align}
u_{0}(x)  &  =L^{-1}\left(  \frac{L(L_{0})}{L_{0}}x\right)  -x\\
v_{0}(x)  &  =y_{0}\left(  x+u_{0}\right)
\end{align}
where $L\left(  z\right)  =\int_{0}^{z}\sqrt{1+\left(  y_{0}^{\prime}\left(
\zeta\right)  \right)  ^{2}}d\zeta$ is the length function defined over
$[0,L_{0}]$. Further, $u_{0}(x)$ and $v_{0}(x)$ automatically satisfy the
boundary conditions.
\end{theorem}

\begin{proof}
Based on the properties of $y_{0}(x)$ and the form of the integrand, the
length function%
\begin{equation}
L(z)=%
{\displaystyle\int\limits_{0}^{z}}
\sqrt{1+\left(  y_{0}^{\prime}\left(  \zeta\right)  \right)  ^{2}}d\zeta\text{
\ \ \ \ \ }z\in\lbrack0,L_{0}]
\end{equation}
is continuous and strictly increasing. Therefore, the inverse function
$L^{-1}(z)$ exists. Since the tension is constant throughout, then so is%
\begin{equation}
\left(  \frac{ds}{dx}\right)  _{0}=\frac{L(L_{0})}{L_{0}}%
\end{equation}
where $L(L_{0})$ is the initial length. From the first of equations
\ref{cossin} one gets%
\begin{equation}
\sqrt{1+\left(  y_{0}^{\prime}\left(  x+u_{0}\right)  \right)  ^{2}}d\left(
x+u_{0}\right)  =\left(  \frac{ds}{dx}\right)  _{0}dx
\end{equation}
Now, by integrating both sides, we have%
\begin{equation}
L(x+u_{0})=\frac{L(L_{0})}{L_{0}}x
\end{equation}
with the boundary condition at $x=0$ applied. Then%
\begin{equation}
u_{0}\left(  x\right)  =L^{-1}\left(  \frac{L(L_{0})}{L_{0}}x\right)  -x
\end{equation}
where $L^{-1}(z)$ is the inverse function for $L(z)$.

One can see from Figure \ref{gen2DMot} that a point originally at $x$ is moved
horizontally to $x+u_{0}$ and then lifted vertically by an amount of $v_{0}$,
ending up at $y_{0}\left(  x+u_{0}\right)  $. Therefore, it follows that
$v_{0}\left(  x\right)  =y_{0}\left(  x+u_{0}\right)  $.

The boundary conditions on $u_{0}$ are automatically satisfied: $u_{0}%
(0)=L^{-1}\left(  0\right)  -0=0$ and
\begin{align}
u_{0}(L_{0})  &  =L^{-1}\left(  \frac{L(L_{0})}{L_{0}}L_{0}\right)  -L_{0}\\
&  =L^{-1}\left(  L(L_{0})\right)  -L_{0}=0
\end{align}
Also, $v_{0}=y_{0}\left(  x+u_{0}\right)  $ yields%
\begin{align*}
v_{0}(0)  &  =y_{0}\left(  0+u_{0}(0)\right)  =y_{0}\left(  0\right)  =0\\
v_{0}(L_{0})  &  =y_{0}\left(  L_{0}+u_{0}(L_{0})\right)  =y_{0}\left(
L_{0}\right)  =0
\end{align*}

\end{proof}

\subsection{Analytical verification of Theorem \ref{theo1}}

Before proceeding to numerical studies, a verification of the above theorem
can be presented. Consider the singly kinked initial condition as depicted in
Figure \ref{snglKink2Cases}(b), where tension equalization is assumed.
Clearly,%
\begin{equation}
\left(  \frac{ds}{dx}\right)  _{0}=\frac{\sqrt{a^{2}+h^{2}}+\sqrt{\left(
L_{0}-a\right)  ^{2}+h^{2}}}{L_{0}}%
\end{equation}
We first derive the forms of the initial displacements using the geometry
only. For this, it is sufficient to determine which point $(x_{1},0)$ is
mapped to $(a,h)$. From the figure is it seen that $x_{1}$ must be such that%
\begin{equation}
\frac{\sqrt{a^{2}+h^{2}}}{x_{1}}=\left(  \frac{ds}{dx}\right)  _{0}%
=\frac{L(L_{0})}{L_{0}}%
\end{equation}%
\begin{equation}
x_{1}=\frac{\sqrt{a^{2}+h^{2}}}{\sqrt{a^{2}+h^{2}}+\sqrt{\left(
L_{0}-a\right)  ^{2}+h^{2}}}L_{0}%
\end{equation}
Now, by Proposition \ref{prop1}, $u_{0}$ and $v_{0}$ are linear functions in
force-free segments. Hence,%
\begin{equation}
u_{0}\left(  x\right)  =\left\{
\begin{array}
[c]{ccc}%
\frac{a-x_{1}}{x_{1}}x &  & 0\leq x\leq x_{1}\\
\frac{a-x_{1}}{L_{0}-x_{1}}\left(  L_{0}-x\right)  &  & x_{1}\leq x\leq L_{0}%
\end{array}
\right.  \label{u0singKink}%
\end{equation}%
\begin{equation}
v_{0}\left(  x\right)  =\left\{
\begin{array}
[c]{ccc}%
\frac{h}{x_{1}}x &  & 0\leq x\leq x_{1}\\
\frac{h}{L_{0}-x_{1}}\left(  L_{0}-x\right)  &  & x_{1}\leq x\leq L_{0}%
\end{array}
\right.  \label{v0SingKink}%
\end{equation}

Next, we show this result using Theorem \ref{theo1}. The initial shape and the
length function are given by%
\begin{equation}
y_{0}\left(  x\right)  =\left\{
\begin{array}
[c]{ccc}%
\frac{h}{a}x &  & x\leq a\\
\frac{h}{L_{0}-a}\left(  L_{0}-x\right)  &  & x\geq a
\end{array}
\right.
\end{equation}

\begin{equation}
L\left(  z\right)  =\left\{
\begin{array}
[c]{ccc}%
s_{1}z &  & z\leq a\\
s_{2}z-\left(  s_{2}-s_{1}\right)  a &  & z\geq a
\end{array}
\right.
\end{equation}
where $z\in\left[  0,L_{0}\right]  $, $s_{1}=\sqrt{1+\left(  \frac{h}%
{a}\right)  ^{2}}$, and $s_{2}=\sqrt{1+\left(  \frac{h}{L_{0}-a}\right)  ^{2}%
}$. The inverse of the length function is%
\begin{equation}
L^{-1}\left(  z\right)  =\left\{
\begin{array}
[c]{ccc}%
\frac{z}{s_{1}} &  & z\leq s_{1}a\\
\frac{z+\left(  s_{2}-s_{1}\right)  a}{s_{2}} &  & z\geq s_{1}a
\end{array}
\right.
\end{equation}
where $z\in\left[  0,L(L_{0})\right]  $. Let
\begin{equation}
L(L_{0})=\sqrt{a^{2}+h^{2}}+\sqrt{\left(  L_{0}-a\right)  ^{2}+h^{2}}=L_{i}%
\end{equation}
Then, by Theorem \ref{theo1},%
\begin{equation}
u_{0}=L^{-1}\left(  \frac{L_{i}}{L_{0}}x\right)  -x
\end{equation}%
\begin{equation}
u_{0}=\left\{
\begin{array}
[c]{ccc}%
\frac{L_{i}/L_{0}}{s_{1}}x-x &  & x\leq\frac{L_{0}}{L_{i}}s_{1}a\\
\frac{L_{i}/L_{0}x-s_{1}a}{s_{2}}+a-x &  & x\geq\frac{L_{0}}{L_{i}}s_{1}a
\end{array}
\right.
\end{equation}
Note that $\frac{L_{0}}{L_{i}}s_{1}a$ is equal to $x_{1}$. After
manipulations, one gets Equation \ref{u0singKink} exactly. The switching point
for $y$, $x=a$, corresponds to that of $x+u_{0}$, which corresponds to that of
$u_{0}$, $x=x_{1}$. Therefore, we simply insert the piecewise definitions in
their corresponding places in $y(x+u_{0})$ to get%
\begin{align}
&  v_{0}=\\
&  \left\{
\begin{array}
[c]{ccc}%
\frac{h}{a}\left(  \frac{a-x_{1}}{x_{1}}x+x\right)  &  & x\leq x_{1}\\
\frac{h}{L_{0}-a}\left(  L_{0}-\left(  \frac{a-x_{1}}{L_{0}-x_{1}}\left(
L_{0}-x\right)  +x\right)  \right)  &  & x\geq x_{1}%
\end{array}
\right.
\end{align}
which, after simplifications, reduces to Equation \ref{v0SingKink}. This
verifies the theorem.

For other types of initial shape functions, obtaining the explicit form of
$u_{0}$ involving simple functions becomes almost impossible or quite
complicated. For example, for an initial sine or cosine shape, $L(z)$ involves
elliptic integrals. Nevertheless, numerical implementation is quite
straightforward, which is used in every example of the next section on
numerical experiments. This clearly demonstrates the usefulness of Theorem
\ref{theo1}.

\subsection{Initial conditions that are compatible with purely transverse
motions}

If $u_{0}\neq0$ at $t=0$, then horizontal motions for $t>0$ are unavoidable
because the related equation of motion becomes $u_{tt}=c^{2}u_{xx}$, for which
$u\left(  x,t\right)  =0$ is not a viable solution since it does not satisfy
the initial conditions. Therefore, in addition to $k=0$, another necessary
condition for purely transverse motions is $u_{0}\left(  x\right)  =0$. This
forces, by Theorem \ref{theo1}, that%
\begin{equation}
L\left(  x\right)  =\frac{L(L_{0})}{L_{0}}x
\end{equation}
This can happen only if the $y_{0}^{\prime}\left(  x\right)  $ is constant in
at most a piecewise manner. If $y_{0}^{\prime}\left(  x\right)  =0$ over the
whole domain then by the virtue of the boundary conditions $y_{0}\left(
x\right)  =0$, which we discard since it results in no motion. Hence,
$y_{0}^{\prime}\left(  x\right)  $ can be constant only in a piecewise manner
over finite regions, the union of which would equal $\left[  0,L_{0}\right]
$. Therefore, a consequence of $u_{0}\left(  x\right)  =0$ is that the initial
shape $y_{0}\left(  x\right)  $ can only have straight segments in a piecewise
continuous manner, simplest case of which is Euler's singly kinked initial
condition. This explains the apparent success of Euler's proposal.

Note that an initial shape with straight segments in a piecewise manner
restricts the initial horizontal forces to be only point forces, a result
which we obtained previously after a different line of reasoning.

The initial shape can be obtained in two ways: a) by allowing tension
equalization, i.e. with zero horizontal constraint force, b) by using initial
horizontal point forces at kinks. In the latter case, as soon as the string is
released the horizontal constraint forces disappear. The discontinuity in the
initial tension around kinks caused by the horizontal forces now creates an
imbalance due to the absence of them just after the release. This will cause
horizontal motions. Therefore, for purely transverse motions it is also
necessary that the tension be allowed to balance, which amounts to zero
horizontal force constraints.

If the initial tension is uniform over the whole string then the slopes of the
straight segments on either side a kink must be equal in magnitude and
opposite in sign in order to preserve force balance in horizontal direction
when the motion starts.\emph{ }It is simple to show that this must be true for
all kinks. That is, the magnitudes of the slopes of all segments must be the
same. Hence, we proved the following.

\begin{theorem}
\label{Theo2}The necessary and sufficient conditions for a string to execute
purely transverse motions are:

\begin{enumerate}
\item $k=\frac{T_{f}}{T_{0}+T_{f}}=0$, i.e. either $E=0$, $A_{f}=0,$ or
$T_{0}\gg T_{f}$.

\item Initial shape must be formed by straight segments around discretely
distributed kink points,

\item There are no initial horizontal force constraints, or, equivalently, the
initial tension is uniform everywhere,

\item Magnitudes of slopes of all segments must be the same and signs of those
on each side of a kink point must be opposite.
\end{enumerate}
\end{theorem}

It is now seen clearly that pure transverse vibrations of an ideal string are
indeed rarities. For a single kink case, the string must be lifted at the
mid-span. Therefore, Euler's proposal is acceptable only when the kink point
is at the middle. Otherwise, pure transverse motions are not possible.

For multiple kinks, the situation becomes interesting. If we start with a
positive slope at $x=0$, without loss of generality, and if each segment is
represented by a vector $\bar{r}_{i}$, $i=1,\cdots,N$, where $N$ is the number
of segments, then we can write down the following vector loop equation.%
\begin{equation}
\sum_{i=1}^{N}\bar{r}_{i}=\sum_{i=1}^{N}r_{i}\left[
\begin{array}
[c]{c}%
\cos\alpha\\
\left(  -1\right)  ^{i+1}\sin\alpha
\end{array}
\right]  =\left[
\begin{array}
[c]{c}%
L_{0}\\
0
\end{array}
\right]
\end{equation}
where $r_{i}$ is the length of $i^{\text{th}}$ segment and $\alpha$ is the
magnitude of the slope angle for all segments. By letting $L_{i}=r_{i}%
\cos\alpha$, projected length on $x$-axis, one gets%
\begin{equation}
\sum_{i=1}^{N}\left[
\begin{array}
[c]{cc}%
L_{i} & \left(  -1\right)  ^{i+1}L_{i}\tan\alpha
\end{array}
\right]  ^{T}=\left[
\begin{array}
[c]{cc}%
L_{0} & 0
\end{array}
\right]  ^{T}%
\end{equation}
which leads to the following scalar equations.%
\begin{align}
\sum_{i=1}^{N}L_{i}  &  =L_{0}\label{SegmentLengths1}\\
\sum_{i=1}^{N}\left(  -1\right)  ^{i+1}L_{i}  &  =0 \label{SegmentLengths2}%
\end{align}
regardless of $\alpha$. Note that the special cases $\alpha=\pm\frac{\pi}{2}$
are not viable. These equations have a unique solution only when $N=2$, namely%
\[
L_{1}=L_{2}=\frac{L_{0}}{2}%
\]
which correspond to the example of a string lifted at the middle point, as was
discussed earlier.

Equations \ref{SegmentLengths1} and \ref{SegmentLengths2} lead to%
\begin{equation}
\sum_{i=1\text{ odds}}^{N}L_{i}=\sum_{i=2\text{ evens}}^{N}L_{i}=\frac{L_{0}%
}{2}%
\end{equation}

For $N>2$ one would have a family of solutions that depend on $N-2$
parameters. Let $L_{1}$ and $L_{2}$ be the dependent variables and $L_{i}$,
$i>2$, be taken as free parameters. Then, the solutions for $L_{1}$ and
$L_{2}$ in terms of the free parameters are%
\begin{align*}
L_{1}  &  =\frac{L_{0}}{2}-\sum_{i=3}^{N}L_{i}\text{ \ \ (}i\text{ is odd)}\\
L_{2}  &  =\frac{L_{0}}{2}-\sum_{i=4}^{N}L_{i}\text{ \ \ (}i\text{ is even)}%
\end{align*}
provided that $L_{i}\geqslant0$ for all $i$.

Note that for any given number of segments, $N$, one gets all others,
$n=2,3,\cdots,N-1$ by setting the lengths of $N-n$ segments to zero.

As a demonstration we consider $N=3$, for which we have the following result.%
\begin{align*}
L_{1}  &  =\frac{L_{0}}{2}-L_{3}\\
L_{2}  &  =\frac{L_{0}}{2}%
\end{align*}
where $L_{3}\in\left[  0,\frac{L_{0}}{2}\right]  $ is the free parameter.
Figure \ref{FigIniShapePureTrans} shows three examples of initial shapes with
two kinks. For all cases the slope angle is $\pm45%
{{}^\circ}%
$. The solid line is when $L_{3}=\frac{1}{4}$ and the dashed line below
corresponds to $L_{3}=0.4$. The dotted-dashed line above is for $L_{3}=0$. It
is seen that the single kink solution is a limit case of two-kink solutions,
as discussed earlier. It is easily seen how all other solutions can be
obtained geometrically. Also note that the points where the shape intersects
the $x-$axis are vibration nodes. Figure \ref{inicondpuretranstriplekink}
illustrates initial shapes with three kinks.%

\begin{figure}
[ptb]
\begin{center}
\includegraphics[scale=0.21]%
{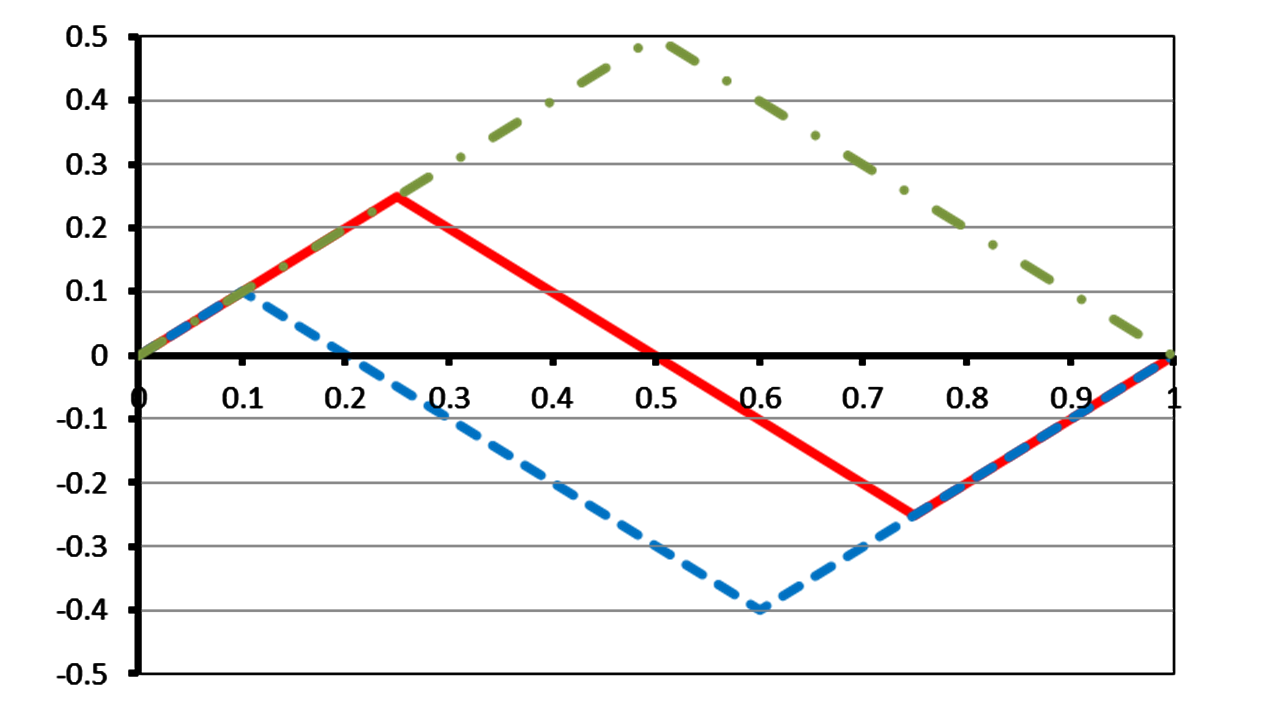}%
\caption{Initial shapes with two kinks and with equal and opposite slopes on
each side of kinks. All slope angles are $\pm45{{}^\circ}$.}%
\label{FigIniShapePureTrans}%
\end{center}
\end{figure}
%

\begin{figure}
[ptb]
\begin{center}
\includegraphics[scale=0.21]%
{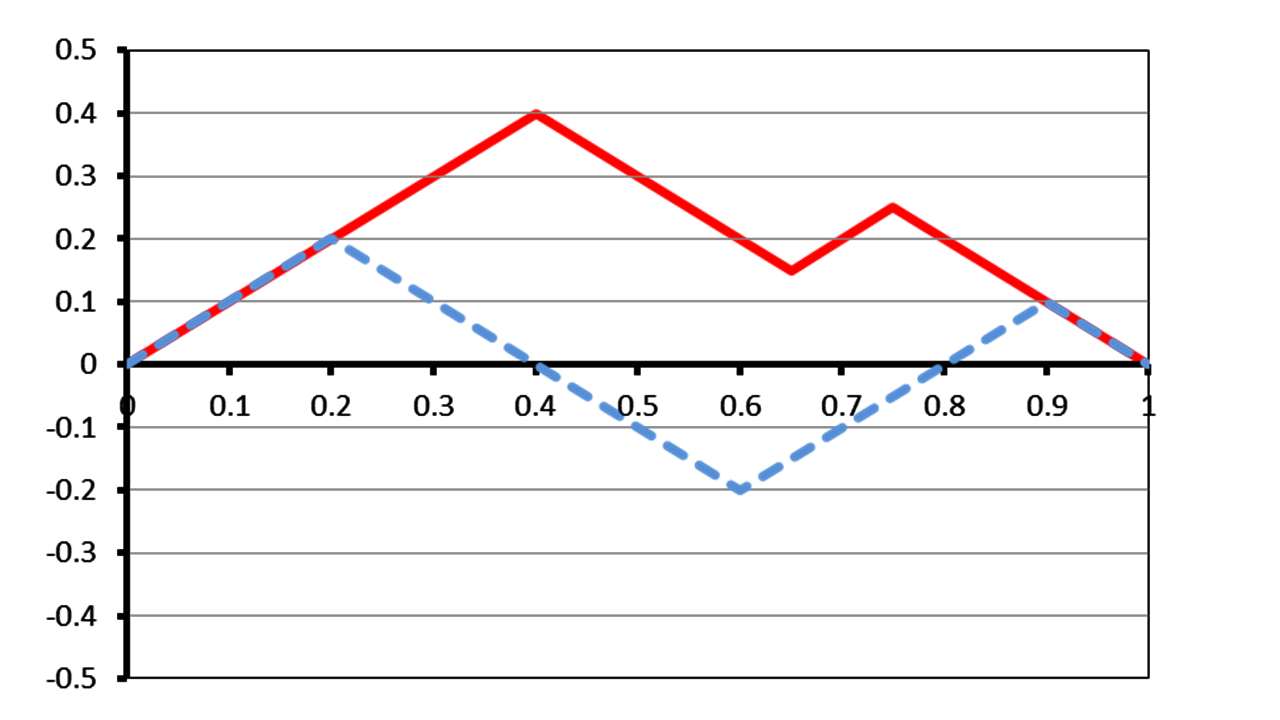}%
\caption{Initial shapes with three kinks and with equal and opposite slopes on
each side of kinks. All slope angles are $\pm45{{}^\circ}$.}%
\label{inicondpuretranstriplekink}%
\end{center}
\end{figure}

Another interesting, yet expected, result pertains to the form of the length
function. Since the slopes of the segments are constant up to a sign
difference the integrand in the definition of the length function becomes%
\begin{equation}
\sqrt{1+\left(  y_{0}^{\prime}\left(  x\right)  \right)  ^{2}}=\sqrt
{1+\tan^{2}\alpha}=\sec\alpha
\end{equation}
regardless of the segment. Therefore, despite having a piecewise defined
shape, the length function is simply given by%
\begin{equation}
L\left(  z\right)  =\left(  \sec\alpha\right)  z \label{specialLengthFunc}%
\end{equation}
for all $z\in\left[  0,L_{0}\right]  $, the inverse of which is
\begin{equation}
L^{-1}\left(  z\right)  =\frac{1}{\sec\alpha}z=\left(  \cos\alpha\right)  z
\label{SpecialLengthFuncInverse}%
\end{equation}
This result can be used to prove the following.

\begin{proposition}
Provided that the initial tension is equalized, the initial horizontal
displacement is zero, $u_{0}\left(  x\right)  =0$, if and only if the initial
shape is formed by straight segments between isolated kinks, on each side of
which the segment slopes are equal and opposite.
\end{proposition}

\begin{proof}
That $u_{0}=0$ implies the initial shape is as declared in the proposition was
proven in Theorem \ref{Theo2}. For the converse, one assumes that the initial
shape is as declared. Then, by Theorem \ref{theo1} (since the tension is
uniform), and, equations \ref{specialLengthFunc} and
\ref{SpecialLengthFuncInverse}, one would have%
\begin{align}
u_{0}  &  =L^{-1}\left(  \frac{L(L_{0})}{L_{0}}x\right)  -x\\
&  =L^{-1}\left(  \frac{\left(  \sec\alpha\right)  L_{0}}{L_{0}}x\right)  -x\\
&  =L^{-1}\left(  \left(  \sec\alpha\right)  x\right)  -x\\
&  =\left(  \cos\alpha\right)  \left(  \sec\alpha\right)  x-x=0
\end{align}
which completes the proof of the proposition.
\end{proof}

This proposition now leads to the following theorem, the proof of which is
straightforward and, therefore, omitted.

\begin{theorem}
\label{PureTransTheorem}Given that $k=0$ and the initial tension is equalized,
an ideal string executes purely transverse motions if and only if the initial
shape is formed by straight segments between isolated kinks, on each side of
which the segment slopes are equal and opposite.
\end{theorem}

In the sequel, the effect of initial conditions is explicitly demonstrated by
numerical experiments. Further, we also present, from the results of actual
experiments, the fact that a string with conditions close to the ideal cannot
help but execute horizontal motions.

\section{NUMERICAL EXPERIMENTS}

In this section, we present numerical analyses of various interesting cases in
order to demonstrate the findings of previous sections. Both the process of
setting up of initial conditions and resulting solutions are discussed. For
verification, the solutions are obtained, whenever possible, on two different
numerical analysis platforms using different methods.

\subsection{Two distinct cases with the same linear shape}

In Figure \ref{snglKink2Cases}(a), the tensions are different in the left and
the right regions. By Proposition \ref{prop1}, the extension ratio in the left
region is $\sqrt{1+\left(  h/a\right)  ^{2}}$. Since points $x=0$ and $x=a$
are mapped to themselves, again by Proposition \ref{prop1}, $u_{0}(x)=0$,
which, after applying the left boundary condition, gives $v_{0,x}=h/a$, or
$v_{0}=(h/a)x$. Similar results are obtained for the right region. Thus, we
have $u_{0}=0$ and%
\begin{equation}
v_{0}=\left\{
\begin{array}
[c]{cc}%
\frac{h}{a}x & 0\leqslant x\leqslant a\\
\frac{h}{L_{0}-a}\left(  L_{0}-x\right)  & a\leqslant x\leqslant L_{0}%
\end{array}
\right.
\end{equation}

The case of Figure \ref{snglKink2Cases}(b), in which the tension is equalized,
was studied earlier and the displacement fields were given in equations
\ref{u0singKink} and \ref{v0SingKink}.

The two problems in Figure \ref{snglKink2Cases} are quite different from each
other, although they are governed by the same equations of motion and the
initial shape of the string looks the same. In Figure \ref{snglKink2Cases}(a),
the points initially move vertically up. However, as soon as the motion starts
the points will also execute horizontal motions due to the imbalance between
the tensions, unless $a=L_{0}/2$. Whereas, in Figure \ref{snglKink2Cases}(b)
the points will certainly undergo horizontal motions, because they are already
displaced horizontally in the initial configuration.%
\begin{figure}
[h]
\begin{center}
\includegraphics[scale=0.39]%
{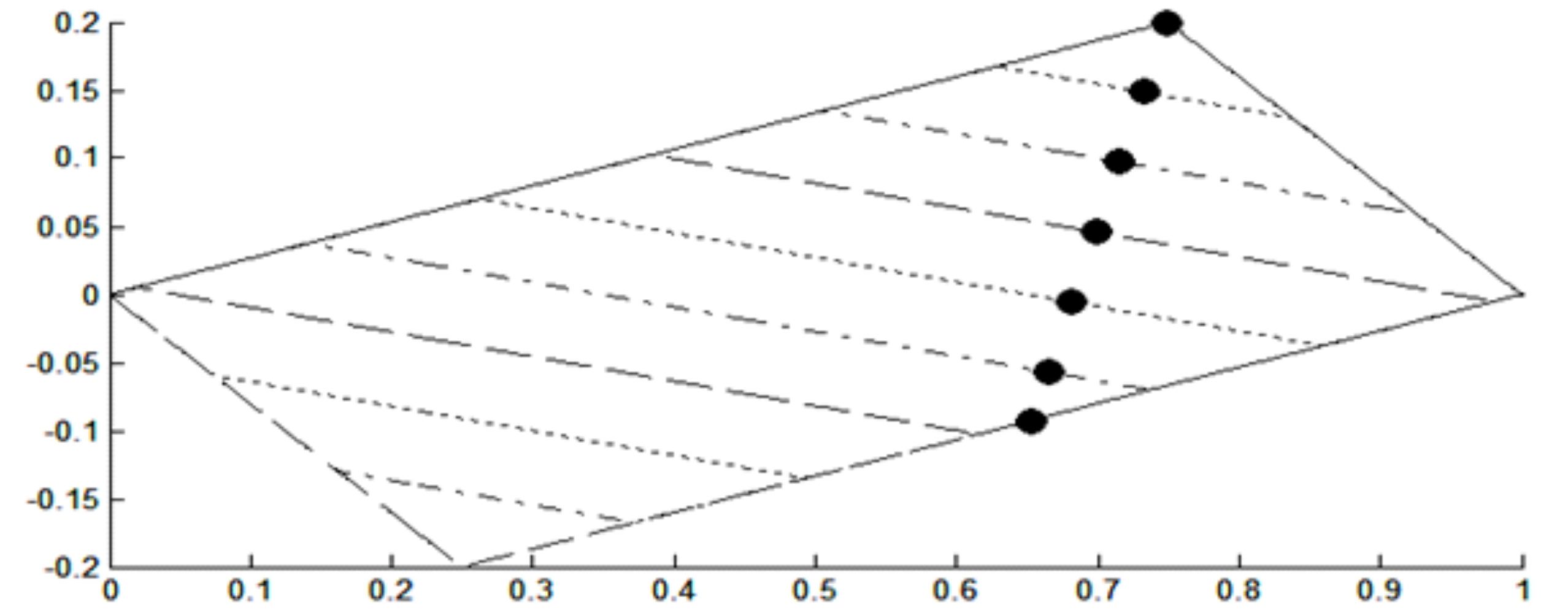}%
\caption{String with an initial single kink in motion. The dark dot is the
path of the kink point. Case: $k=0$.}%
\label{snglKinkMot}%
\end{center}
\end{figure}

Figure \ref{snglKinkMot} shows the motion for case (b) and the trace of the
kink point. All points, except the boundaries, move such that they dwell for a
finite duration when they reach the outer envelopes. Then, they move in the
opposite direction starting with infinite accelerations, a behavior that can
be easily observed using Euler's solutions. This is an unavoidable consequence
of the chosen material model.%
\begin{figure}
[h]
\begin{center}
\includegraphics[scale=0.45]%
{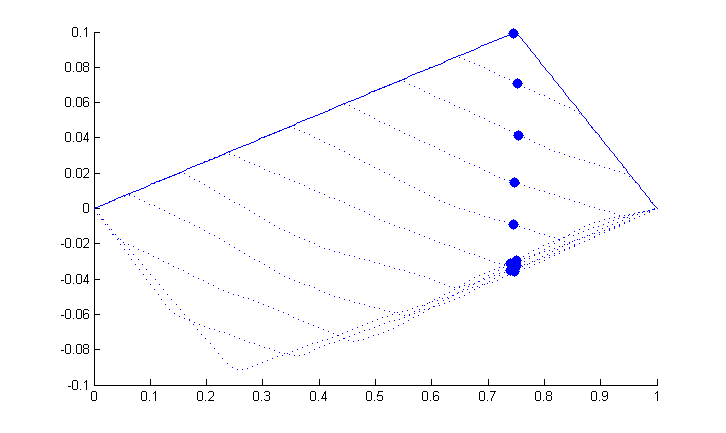}%
\caption{A singly kinked string solution. The kink point was initially
constrained to move vertically only. Case: $k>0$.}%
\label{singkinkclas}%
\end{center}
\end{figure}
\begin{figure}
[h]
\begin{center}
\includegraphics[scale=0.45]%
{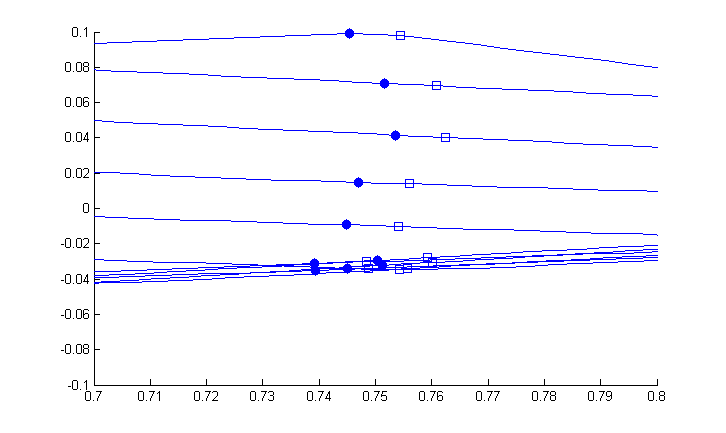}%
\caption{Detail around the single kink (initially constrained). The points
first move towards right due to higher tension, then towards left due to
horizontal oscillation.}%
\end{center}
\end{figure}
\begin{figure}
[h]
\begin{center}
\includegraphics[scale=0.5]%
{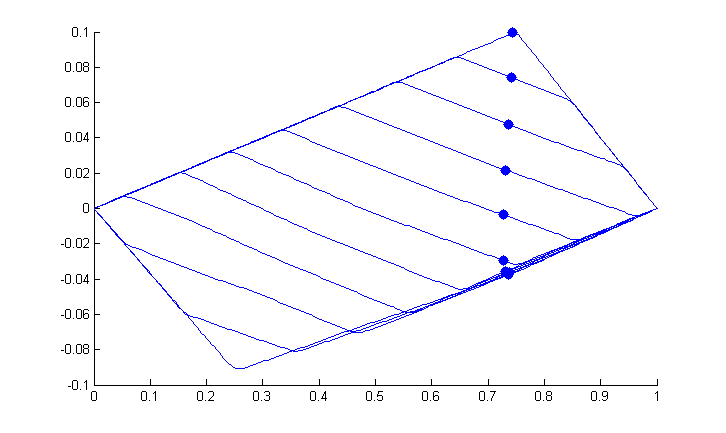}%
\caption{A singly kinked string solution. The tension is equalized initially.
The final kink point originates from a point in the left region. Case: $k>0$.}%
\label{singknknew}%
\end{center}
\end{figure}
\begin{figure}
[h]
\begin{center}
\includegraphics[scale=0.45]%
{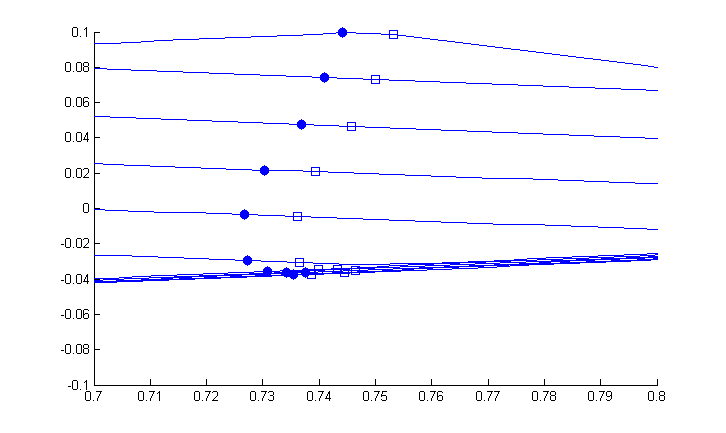}%
\caption{Detail around the single kink (initially relaxed). The points first
move towards left due to initial horizontal force field, then towards right
due to horizontal oscillation.}%
\label{singknknewdetail}%
\end{center}
\end{figure}

\subsection{Linear shape with a double kink}%

\begin{figure}
[h]
\begin{center}
\includegraphics[scale=0.3]%
{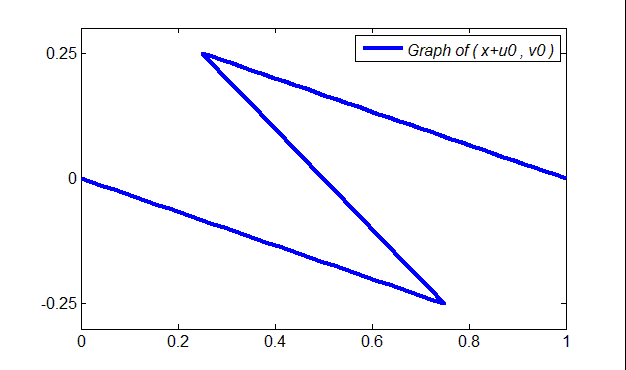}%
\caption{The graph of $(x+u_{0},v_{0})$ for the doubly kinked string.}%
\label{baddblkinkgraph}%
\end{center}
\end{figure}
\begin{figure}
[h]
\begin{center}
\includegraphics[scale=0.3]%
{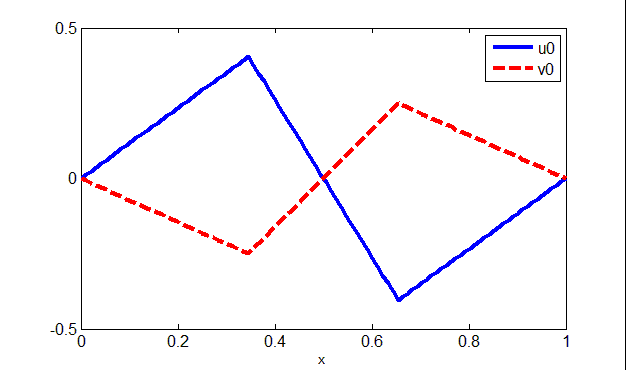}%
\caption{Displacement fields for a doubly kinked string.}%
\label{baddblkink}%
\end{center}
\end{figure}
%

\begin{figure}
[h]
\begin{center}
\includegraphics[scale=0.35]%
{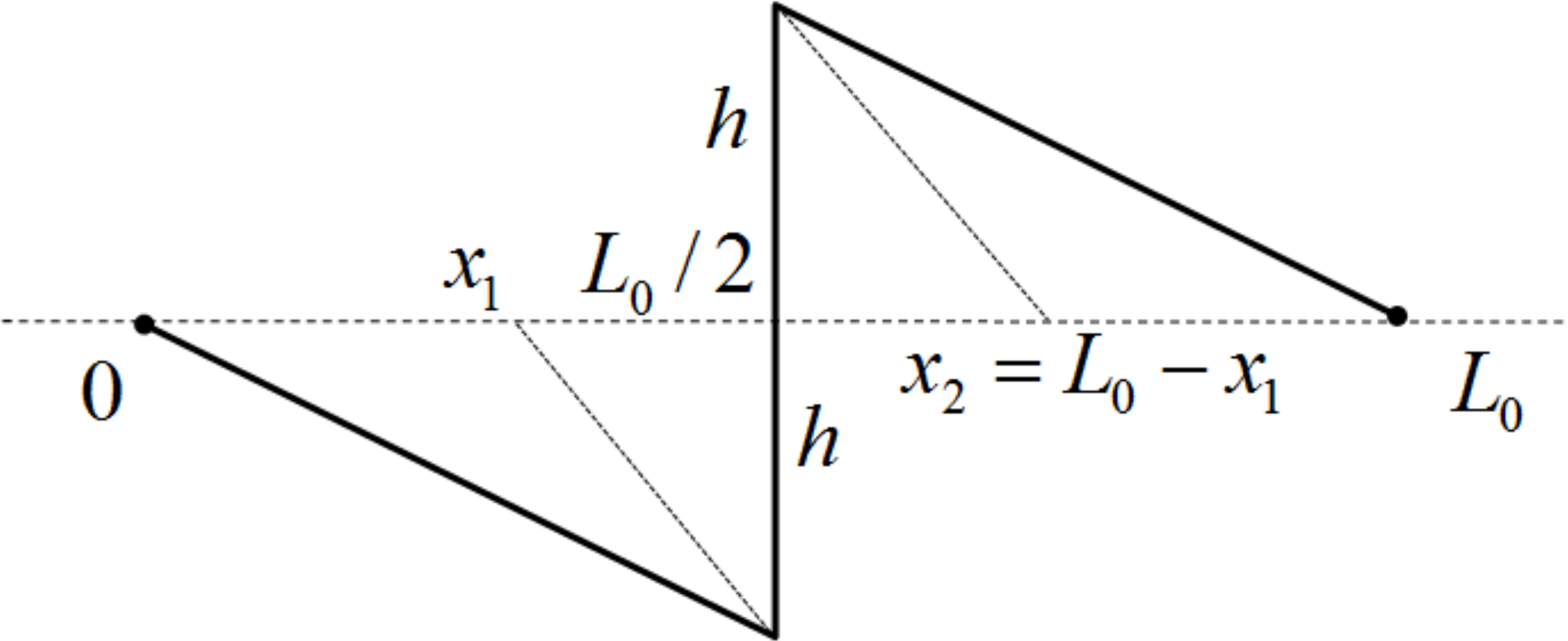}%
\caption{String with an initial double kink and infinite slope at midpoint.}%
\label{dblKinkIC}%
\end{center}
\end{figure}

The doubly kinked initial shape shown in Figure \ref{baddblkinkgraph}\ is
presented as another interesting case, which the classical model has no way of
representing. A general case of initial displacement gradients is shown in
Figure \ref{baddblkink}. A particular case is given in Figure \ref{dblKinkIC}.
The string is allowed to equalize the tension everywhere. Therefore, the
extension ratio is constant everywhere in the string. Performing similar
analyses as was done above, we obtain the following initial conditions.%
\begin{equation}
u_{0}\left(  x\right)  =\left\{
\begin{array}
[c]{ccc}%
\frac{\frac{L_{0}}{2}-x_{1}}{x_{1}}x &  & 0\leq x\leq x_{1}\\
\frac{L_{0}}{2}-x &  & x_{1}\leq x\leq L_{0}-x_{1}\\
-\frac{\left(  \frac{L_{0}}{2}-x_{1}\right)  \left(  L_{0}-x\right)  }{x_{1}}
&  & L_{0}-x_{1}\leq x\leq L_{0}%
\end{array}
\right.
\end{equation}%
\begin{equation}
v_{0}\left(  x\right)  =\left\{
\begin{array}
[c]{ccc}%
-\frac{h}{x_{1}}x &  & 0\leq x\leq x_{1}\\
\frac{h\left(  2x-L_{0}\right)  }{L_{0}-2x_{1}} &  & x_{1}\leq x\leq
L_{0}-x_{1}\\
\frac{h}{x_{1}}\left(  L_{0}-x\right)  &  & L_{0}-x_{1}\leq x\leq L_{0}%
\end{array}
\right.
\end{equation}
%

\begin{figure}
[h]
\begin{center}
\includegraphics[scale=0.5]%
{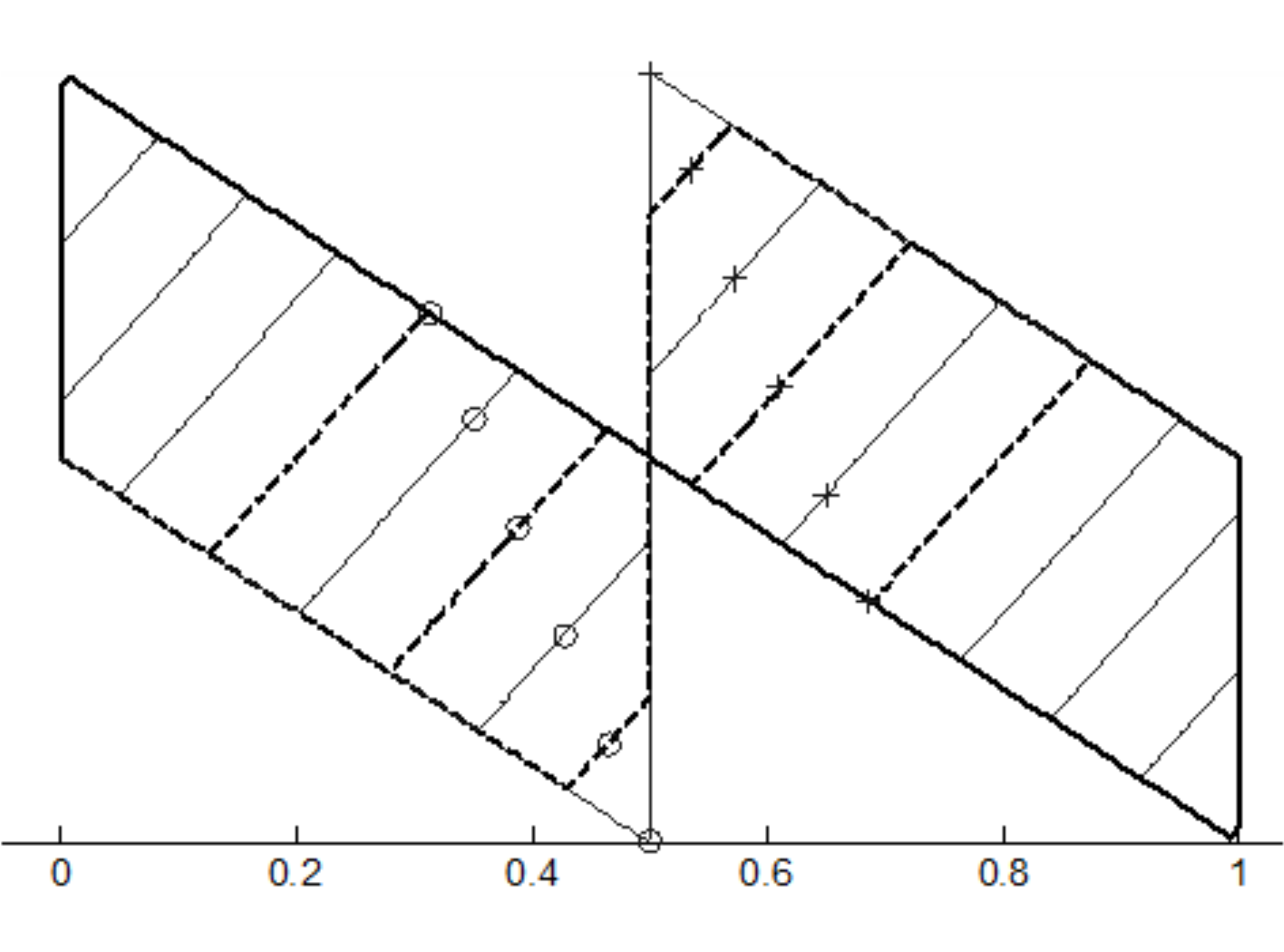}%
\caption{Motion of initially doubly kinked string. Case: $k=0$.}%
\label{dblKinkInMot}%
\end{center}
\end{figure}

Figure \ref{dblKinkInMot} shows the results of numerical analysis for half a
period. The thick solid line is the shape at the end of half a period. The
motion closely follows Euler's solutions. All points move along parallel,
non-vertical lines and exhibit the characteristic dwelling on envelopes. The
midpoint is a vibration node.%

\begin{figure}
[h]
\begin{center}
\includegraphics[scale=0.44]%
{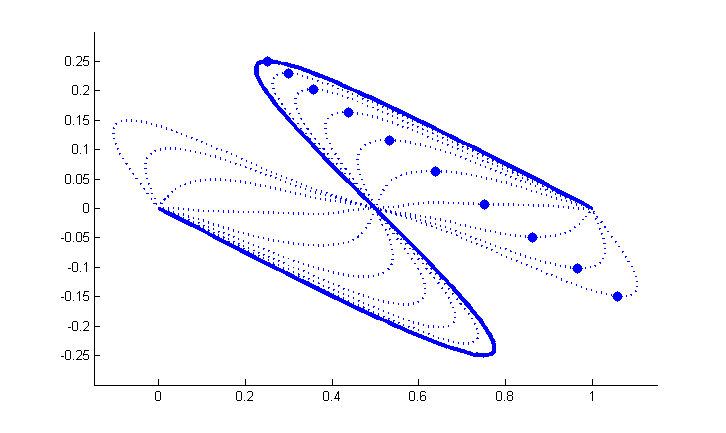}%
\caption{A smoothed double kink solution. The initial conditions were
$u_{0}=0.5\sin(2\pi x)$ and $v_{0}=-0.25\sin(2\pi x)$.}%
\label{baddblkink2}%
\end{center}
\end{figure}

Figure \ref{baddblkink2} shows the solutions to initial conditions
approximating a sharp double kink with a smoothed one. The result is obtained
with surprisingly simple functions for $u_{0}$ and $v_{0}$.

\subsection{Non-taut string}

\begin{center}%
\begin{figure}
[h]
\begin{center}
\includegraphics[scale=0.3]%
{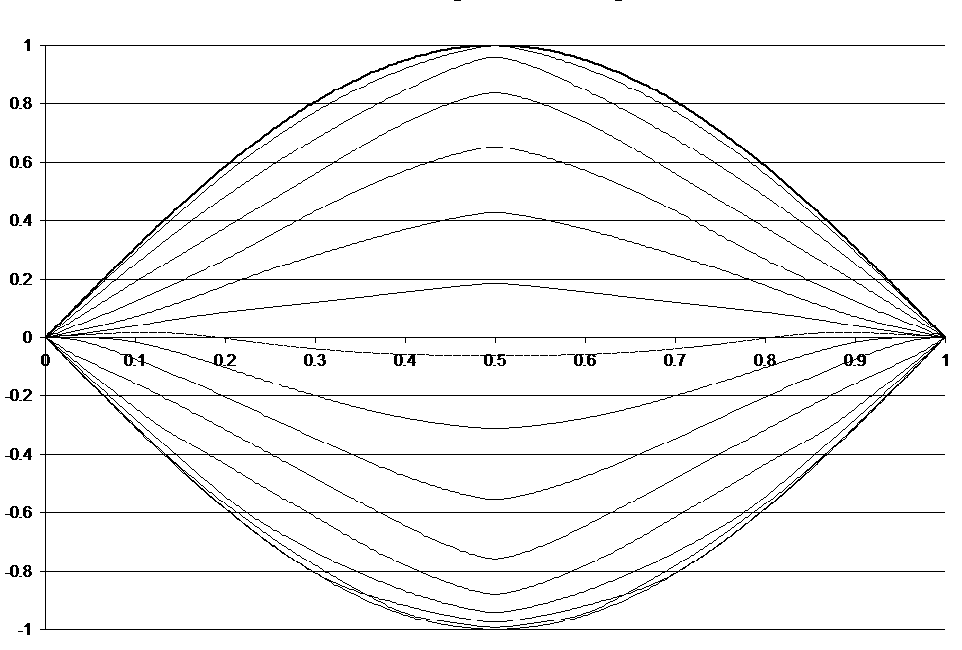}%
\caption{Motion of a string that was non-taut and non-slack at rest.}%
\label{nonTaut}%
\end{center}
\end{figure}

\end{center}

Figure \ref{nonTaut} shows a numerical solution to the general equation for an
initially non-taut, non-slack string. The initial conditions were a half sine
wave for the shape and zero initial velocity. The numerical solver seems to be
reliable since it approaches well to a half sine wave in the negative region
even after numerous iteration steps in time. The shape of the curve near zero
line deviates significantly from a sine function, which, however, is perfectly
recovered at the end of \ half period.

\subsection{Twisted loop: a self-intersecting shape}%

\begin{figure}
[h]
\begin{center}
\includegraphics[scale=0.5]%
{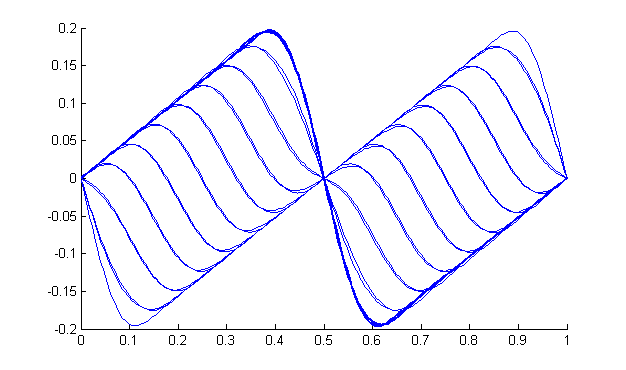}%
\caption{Horizontal motion for the twisted loop initial condition. The darker
curve corresponds to $u_{0}$ and $u(x,T/2)$, where $T$ is the period.}%
\label{twistloopu0}%
\end{center}
\end{figure}
%

\begin{figure}
[h]
\begin{center}
\includegraphics[scale=0.5]%
{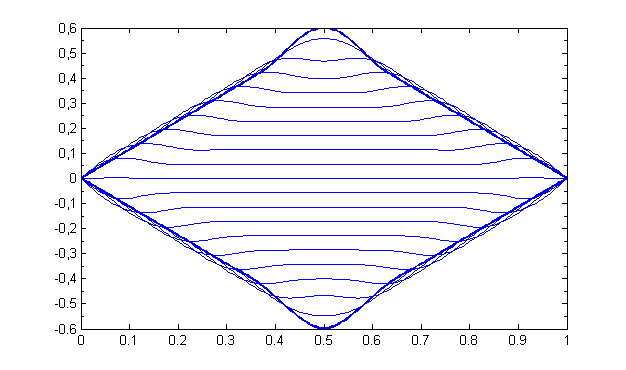}%
\caption{Vertical motion for the twisted loop initial condition. The darker
curve corresponds to $v_{0}$ and $v(x,T/2)$, where $T$ is the period}%
\label{twistloopv0}%
\end{center}
\end{figure}

Another interesting initial condition is obtained by the initial conditions
depicted in figures \ref{twistloopu0} and \ref{twistloopv0}. The graph of
these is given in Figure \ref{twistloopgraph}. The string is allowed to
self-intersect. It is interesting that the period of $u(x,t)$ is half those of
$v(x,t)$ and the graphed motion. The formation of kinks in the string shape is
due to tangent points of lines with the circle in the initial conditions, at
which there are curvature discontinuities. These solutions are based on the
case of $k=0$. The Euler solutions for individual wave equations and finite
element based solutions were both obtained for comparison. They were identical
up to numerical accuracies.%
\begin{figure}
[h]
\begin{center}
\includegraphics[scale=0.5]%
{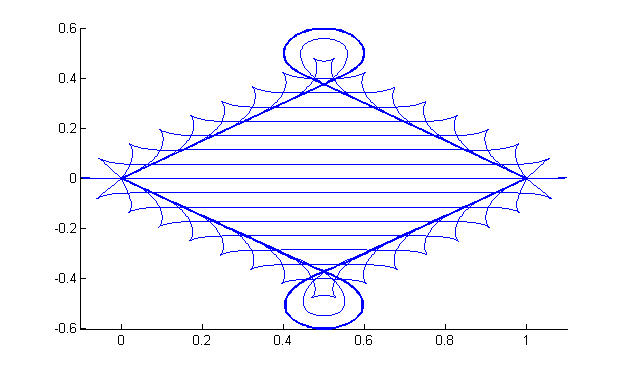}%
\caption{2D motion for the twisted loop initial condition. The darker curves
correspond to the initial shape and the shape after half a period.}%
\label{twistloopgraph}%
\end{center}
\end{figure}

\section{HIGH-SPEED CAMERA EXPERIMENTS}

In order to demonstrate the soundness of the theory and the numerical
investigations, a real experiment is conducted using everyday items, except
for a high-speed camera.

For the string a common round elastic string with latex fibers is used, which
proved to be less than ideal due to its high internal damping. Nevertheless,
the first few milliseconds of the motion were sufficient to make a point, at
least qualitatively.

The string was stretched between two strong posts approximately 2500 mm apart
as measured from the boundaries of the string. A white marker is applied to a
small portion of the string, 843 mm from the left end (with respect to the
camera), which was used to trace the motion of a material point. The string
was pulled up and dragged horizontally away from the left end such that it was
displaced by about 100 mm in both directions. Then, the camera was turned on
and the string was released.

The camera used was SpeedCam (MiniVis \#00655 v1.7.36) with a resolution of
640x512 pixels. The video capture speed was 200 Hz (fps), corresponding to a
period of 5 ms.

The resulting image frames were then analyzed using image processing
techniques. The first 21 frames after the start of motion were processed. The
first frame was used as a substrate onto which the consequent ones were
superposed after removing the string image except the marker (Figure
\ref{CamResult}).%

\begin{figure}
[h]
\begin{center}
\includegraphics[scale=0.5]%
{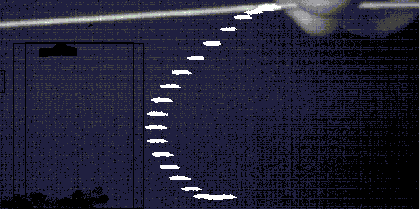}%
\caption{Motion of the marker in the first 105 ms.}%
\label{CamResult}%
\end{center}
\end{figure}

A full motion trace was also performed in order to indicate the effect of
damping (Figures \ref{FullTraceExp}\ and \ref{FullTraceNum}). The horizontal
motion decayed much faster then the vertical motion, due to its faster speed.

In order to be able to compare the experimental results to those obtained from numerical analyses viscous damping terms were introduced to both PDEs. The damping ratios and coefficient were estimated from the camera experiments. Under these conditions the results of the experiment seemed to be quite similar to those of the
numerical except for the fact that the real string followed a smoother path.
This is due to the effect of bending that was neglected in
numerical experiments.%

\begin{figure}
[h]
\begin{center}
\includegraphics[scale=0.4]%
{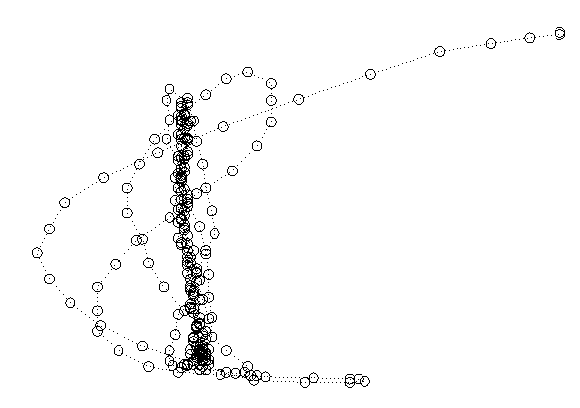}%
\caption{Full trace of the marker in camera experiment.}%
\label{FullTraceExp}%
\end{center}
\end{figure}
%

\begin{figure}
[h]
\begin{center}
\includegraphics[scale=0.5]%
{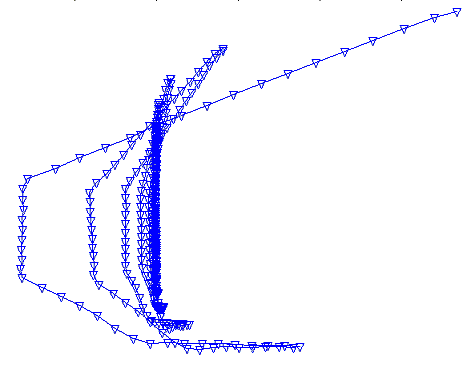}%
\caption{Full trace of the marker in numerical solutions.}%
\label{FullTraceNum}%
\end{center}
\end{figure}

As for the other parameters, the following are measured, albeit roughly.%

\begin{tabular}
[c]{rl}%
Diameter: & $2.6$ mm\\
Linear density: & $4.14\times10^{-3}$ kg/m\\
Volumetric density: & $779$ kg/m$^{3}$\\
Spring constant: & 30 N/m\\
Elastic Modulus: & 3.5 MPa\\
$T_{f}$ and $T_{0}$: & 18.5 N and 2.3 N\\
$k$: & 0.89\\
$c_{f0}$: & 71 m/s
\end{tabular}

The frequency of horizontal motion is measured as 10.5 Hz, whereas that for
the vertical was about 5.2 Hz. The horizontal motion was almost twice as fast
as the vertical, as was observed in numerical experiments. Whether or not this
was a numerical artifact could not be ascertained.

\section{CONCLUSION}

It is shown that the classical vibrating string model is not defendable due to
lack of accounting for horizontal motions, which are unavoidable except in
extremely rare situations. A new 2D motion model, as a pair of non-linear
PDEs, is developed in order to account for such shortcomings. Model allows a
variety of initial conditions that have no counterparts in the classical
model. The most important features of the model are: a) finite displacements,
b) variable tension and length, c) reduction to the classical model and others
in limit cases, d) handling of previously indescribable cases of
non-taut/non-slack strings, self-intersecting shapes, and so on.

Another new result is the determination of initial displacement gradients for
strings that are initially allowed to have uniform tension. It is also shown
that pure transverse motions are possible only when the string has zero
cross-sectional area or zero elastic modulus, the initial tension is uniform
(i.e. zero initial horizontal forces), and the initial shape is formed by
straight segments around kink points such that the segment slopes on each side
of kinks are equal and opposite.

Various numerical experiments are conducted on different platforms for
verification of the analytic models, which seem to agree with each other to a
high degree of accuracy. High-speed camera experiments were conducted that
seemed to unequivocally support the new theory.

\pagebreak


\bigskip

\section{REFERENCES}

\begin{thebibliography}{99}                                                                                               %


\bibitem {Dalembert1}D'Alembert, J. L. (1747). Recherches Sur la Courbe Que
Forme Une Corde Tendu\"{e} Mise en Vibration. Hist. de l'Acad. Roy. de Berlin,
3, pp. 214-219.

\bibitem {Dalembert2}D'Alembert, J. L. (1747). Suite des Recherches Sur la
Courbe Que Forme Une Corde Tendu\"{e} Mise en Vibration. Hist. de l'Acad. Roy.
de Berlin, 3, pp. 220-249.

\bibitem {Euler}Euler, L. (1748). Sur la Vibration des Cordes. Hist. de
l'Acad. Roy. de Berlin, 4, pp. 69-85.

\bibitem {BernoulliD1732}Bernoulli, D. (1732/1733). Theoremata de
Oscillationibus Corporum Filo Flexili Connexorum et Catenae Verticaliter
Suspensae. Comentarii Academiae Scientiarum Imperialis Petropolitanae, 6, pp. 108-122.

\bibitem {BernoulliD1740}Bernoulli, D. (1740). De Oscillationibus Composites
Praesertim Iss Quae Flunt in Corporibus Ex Filo Flexili Suspensis. Comentarii
Academiae Scientiarum Imperialis Petropolitanae, 12, pp. 97-108.

\bibitem {BernoulliD1753}Bernoulli, D. (1753). R\'{e}flexions et
\'{E}claircissemens Sur les Nouvelles Vibrations des Cordes. Hist. de l'Acad.
Roy. de Berlin, 9, pp. 147-195.

\bibitem {BernoulliJ1732}Bernoulli, J. (1728). Meditationes de Chordis
Vibrantibus. Comentarii Academiae Scientiarum Imperialis Petropolitanae. 3,
pp. 13-28. (Translated by NASA: \textit{Deliberations On Oscillating Strings},
NASA TT F-9515, 1965)

\bibitem {Fourier}Fourier, J. B. (1819). Th\'{e}orie du Mouvement de la
Chaleur Dans les Corps Solides. M\'{e}m. de l'Acad. Roy. des Sci. de l'Inst.
de France, 4, pp. 185-556.

\bibitem {Dirichlet}Dirichlet, P. G. (1829). Sur la Convergence des S\'{e}ries
Trigonom\'{e}triques Qui Servent \`{A} Repr\'{e}senter Une Fonction Arbitraire
Entre des Limites Donn\'{e}es. J. f\"{u}r Math., 4, 157-169. Retrieved from arXiv:0806.1294v1.

\bibitem {Wheeler}Wheeler, G. F. and Crummett, W. P. (1987). The Vibrating
String Controversy. Am. J. Phys., 55(1), 33-37.

\bibitem {Zeeman}Zeeman, E. C. (1993). Controversy in Science: on the Ideas of
Daniel Bernoulli and Ren\'{e} Thom. Nieuw Arch. Wisk. (4), 11(3), pp. 257-282.

\bibitem {Armstead2006}Armstead, D. C. and Karls, M. A. (2006). Does the Wave
Equation Really Work?, PRIMUS, 16 (2), pp. 162-177.
{\small http://dx.doi.org/10.1080/10511970608984144}.

\bibitem {Benson2007}Benson, D. J., Music: A Mathematical Offering, 2007,
Cambridge University Press, UK and USA.

\bibitem {Byron1992}Byron, F. W. and Fuller, R. W., Mathematics of Classical
and Quantum Physics, 1992, Dover Pub. (Original publications, 1969, 1970:
General Publishing Co., Canada; and Constable and Company, UK).

\bibitem {Ciblak2013}Ciblak, N. (2013). Arbitrarily Large Motion Model For
Transverse Vibrations of a String. arXiv:1310.1019v1.

\bibitem {Gottlieb1990}Gottlieb, H. P. W., (1990), Non-linear Vibration of a
Constant-Tension String, J. Sound Vib., v. 143, pp 455-460.

\bibitem {Kreysig1983}Kreysig, E., Advanced Engineering Mathematics, 1983, New
York: John Wiley and Sons, 5th ed.

\bibitem {Lagrange}Lagrange, J.-L. (1759). Recherches Sur la Nature et la
Propagation du Son. Miscell. Taurin., 1, pp. 39-148.

\bibitem {Lai2008}Lai, S. K., Xiang, Y., et al., 2008, Higher-Order
Approximate Solutions For Non-linear Vibration of a Constant-Tension String,
J. Sound Vib., v.317, pp 440-448.

\bibitem {Powers}Powers, D. L., Boundary Value Problems. New York: Saunders
College, 3rd ed.

\bibitem {Rao2002}Rao, G. V., 2002, Moderately Large Amplitude Vibrations of a
Constant Tension String: a Numerical Experiment, J. Sound Vib., v. 263, pp 227-232.

\bibitem {Sagan1961}Sagan, H., Boundary and Eigenvalue Problems in
Mathematical Physics, 1989, Dover Pub. (Original publication, 1961, Wiley)

\bibitem {Taylor}Taylor, B. (1713). De Motu Nervi Tensi. Phil. Trans., 28, pp. 26-32.

\bibitem {Weinstock1974}Weinstock, R., Calculus of Variations, 1974, Dover
Pub. (Original publications, 1952: General Publishing Co., Canada; and,
Constable and Company, UK)
\end{thebibliography}
\end{document}